\documentclass[11pt]{article}
\usepackage{amsfonts,mathtools}
\usepackage{amsmath}
\usepackage{amssymb}
\usepackage{enumerate}
\usepackage{natbib}
\usepackage{color}
\usepackage[hmargin=1in,height=9.1in]{geometry}
\usepackage{amsthm, times}
\usepackage{delarray}
\usepackage{setspace}
 \usepackage[subnum]{cases}
\usepackage{bm,mathrsfs}
\usepackage{multirow}
\usepackage{ifpdf}
\ifpdf
\usepackage[pdftex]{graphicx}
 \else
\usepackage[dvips]{graphicx}
 \fi

\definecolor{webgreen}{rgb}{0,.5,0}
\definecolor{webbrown}{rgb}{.8,0,0}
\definecolor{emphcolor}{rgb}{0.95,0.95,0.95}
\usepackage{hyperref}
\hypersetup{%
          colorlinks=true,
          linkcolor=webbrown,
          filecolor=webbrown,
          citecolor=webgreen,
          breaklinks=true}
\ifpdf \hypersetup{pdftex,
            pdfstartview=Fit, 
            bookmarksopen=true,
            bookmarksnumbered=true
} \else \hypersetup{dvips} \fi

\renewcommand{\S}{\mathcal{S}}
\newcommand{\setT}{\mathcal{T}}

\newcommand{\xq}{\overline{X}_{\mathbf{e}_q}}
\newcommand{\psival}{\psi_\alpha^+(-1)}

\newtheorem{thm}{Theorem}[section]
\newtheorem{prop}{Proposition}[section]

\newtheorem{cor}{Corollary}[section]
\newtheorem{remark}{Remark}[section]
\newtheorem{lem}{Lemma}[section]
\newtheorem{lemma}{Lemma}[section]

\newtheorem{assump}{Assumption}[section]
\newtheorem{definition}{Definition}[section]

\newcommand {\R}{\mathbb{R}}

\newcommand {\F}{\mathcal{F}}

\newcommand {\p}{\mathbb{P}}
\newcommand {\E}{\mathbb{E}}

\newcommand{\be}{\begin{equation}}
\newcommand{\ee}{\end{equation}}

\newcommand{\bq}{\begin{eqnarray}}
\newcommand{\eq}{\end{eqnarray}}
\newcommand{\ex}{\mathbb{E}}

\newcommand{\bs}{\bigskip}

\newcommand{\ind}{1\hspace{-2.1mm}{1}} 

\newcommand{\pr}{\mathbb{P}}

\newcommand{\ep}{\mathbf{e}_\alpha}

\newcommand{\diff}{{\rm d}}

\newcommand{\nn}{\nonumber}

\newcommand{\lev}{L\'{e}vy }

\newcommand{\exph}{\ex^{\Phi(\alpha)}}
\newcommand{\mte}{\mathrm{e}}
\newcommand{\setN}{\{1,\cdots,n\}}

\title{ Optimal Multiple  Stopping with Negative Discount Rate and  Random Refraction Times   under L\'evy Models}
\author{Tim Leung\thanks{IEOR Department, Columbia University, New York, USA. \mbox{E-mail: \texttt{leung@ieor.columbia.edu}}.} \and Kazutoshi Yamazaki\thanks{Department of Mathematics,
Kansai University, Osaka, Japan. E-mail: \mbox{\texttt{kyamazak@kansai-u.ac.jp}}. } \and Hongzhong Zhang\thanks{Statistics Department, Columbia University, New York, USA. \mbox{E-mail: \texttt{hzhang@stat.columbia.edu}.}}}

\date{\today}

\begin{document}

\maketitle
 
\begin{abstract}This paper studies a class of optimal multiple stopping problems driven by  L\'evy processes.   Our model allows for   a \emph{negative}  effective  discount rate,  which arises  in a number of  financial applications, including  stock loans and real options, where the strike price can potentially  grow  at a higher rate  than the original discount  factor.   Moreover, successive  exercise  opportunities  are separated  by   i.i.d. random \emph{refraction times}. Under a wide class of  two-sided \lev models with a general random refraction time,  we rigorously show that the optimal strategy to exercise successive  call options     is uniquely characterized by a  sequence of up-crossing times. The corresponding optimal thresholds are determined explicitly in the single stopping case and recursively in the multiple stopping case. 

\end{abstract}

 \noindent \small{\textbf{Key words:} optimal multiple stopping; negative discount rate; random refraction times; \lev processes; stock loan
 }\\

 \section{Introduction}
We  study a class of optimal multiple stopping problems driven by an underlying  L\'evy process. Two key features of our model are that (i)  the discount rate can be negative or positive, and (ii) the sequence of  admissible stopping times are separated  by   i.i.d. random \emph{refraction periods}. The negative   effective discount rate  is relevant to a number of financial applications. For example, Xia and Zhou \cite{xiazhou} propose a   valuation model for  a stock loan, where  the loan interest rate is higher than the risk-free interest rate. As a result,  the stock loan can be viewed as an American call  option   with a negative effective  discount rate.    An example from the real option  literature \cite{DixitPindyck,McDonald1985}  is when the cost of investment grows at a higher rate  than the firm's discount  rate. Moreover, while the nominal short rate cannot be negative, the real interest rate can potentially  be negative, especially  during low-yield  regimes, according to Black \cite{black95} and references therein. Therefore, extending  the discount rate to the negative domain also enables the evaluation of cash flows under the real interest rate.  

 In the aforementioned applications,   the same option can be exercised repeatedly in the future, meaning  that  an  investor can acquire a series of stock loans, or a firm can make an investment sequentially over time.  This motivates  us to incorporate multiple  stopping opportunities in our analysis. The features of refraction periods and multiple exercises also arise in the pricing of  swing options commonly used for  energy delivery.  For instance,  Carmona and Touzi \cite{Carmona2008} formulate   the valuation of a swing put option  as    optimal multiple stopping problem, with constant refraction periods, under  the  geometric Brownian motion model. In a related study,  Zeghal and Mnif \cite{ZeghalSwing} value a  perpetual American swing put when the underlying L\'evy price process has  no negative jumps. They provide mathematical characterization and numerical solutions to the associated optimal multiple stopping problem. In contrast, we  consider the successive  exercises of  a swing call  option with random refraction times under positive or negative discount rate. We also provide a rigorous analysis of the optimal multiple stopping problem under two-sided  \lev models.

Under a wide class of  two-sided \lev models with a general random refraction time, we  show that the optimal    exercises of multiple perpetual  call options are  characterized by a non-increasing sequence of exercise thresholds (see Proposition \ref{prop_mono} and Theorem \ref{main2} below). The corresponding optimal thresholds are determined explicitly in the single stopping case and recursively in the multiple stopping case. Our results extend the stock loans models by Xia and Zhou \cite{xiazhou} as well as Cai and Sun \cite{Cai2014} from their single stock loan to   sequential stock loans, and from a   geometric Brownian motion \cite{xiazhou} and double-exponential jump diffusion \cite{Cai2014} models to a class of general two-sided \lev processes in our paper.  As such, the minimal assumptions on the refraction times and the underlying \lev process prevent the use of model/distribution-specific properties that are amenable for analysis and computation.  We overcome this  challenge  through the use of Laplace transform, change of measure, martingale theory, along with other analytical techniques.  Our analysis allows for the  recursive computation of  the optimal value function as well as all exercise thresholds, thus providing an  alternative to the simulation  approach commonly found in  existing literature for multi-exercise  options (see \citep{Bender,Hambly_MF04}, among others). We also examine the impact of refraction time distribution  on the optimal  exercise thresholds.

 In our model, the   random  refraction times between consecutive exercise opportunities can also be interpreted  as a result of  successive randomization. To this end,  Kyprianou and  Pistorius   \cite{kypri_Canada} apply  fluctuation theory of L\'evy processes to study the method  of maturity randomization (Canadization) for derivatives  pricing. The randomization procedure turns a finite-maturity option into a perpetual one.   Avram et al.  \cite{Avram_2004} consider a number of exit problems of spectrally negative L\'evy processes, and apply them to value  Russian options with a randomized maturity. In contrast,  we consider a problem with  multiple exercise rights,   allowing for a negative  discount rate, as well as negative and  positive jumps for the underlying \lev process.

  The recent work  by Christensen and Lempa \cite{ChristensenLempa13} discuss  an optimal multiple stopping problem driven by a strong Markov process with  i.i.d.  exponential  refraction periods.  Another related work by Christensen et al. \cite{CIJ2013}  study an optimal multiple stopping  problem with random waiting times in terms of a sequence of single stopping problems. They provide an explicit solution to the problem of a  perpetual put option whereby the sequential exercises are refracted by the  first passage times of a geometric Brownian motion.  Compared to their work, our model not only allows for a negative discount rate and general random refraction times, but also  incorporates jumps in the underlying process via a \lev process  with   positive phase-type jumps and negative jumps from  any distribution. As  discussed  in \cite{Asmussen2004},  \lev processes with  phase-type jumps are capable  of approximating  a general  class of  \lev processes. Herein, the  major mathematical challenge is to characterize the optimal exercise strategies given minimal distributional  structures of the \lev  jumps and refraction times.

The  current paper is also relevant  to  the growing number of financial applications  that involve making sequential timing  decisions. Examples include  multiple-exercise options \cite{carmonaDayanik08,Hambly_MF04}, portfolios of employee stock options \citep{ChristensenLempa13,GrasselliHenderson2008,LeungSircarESO_MF09}, sequential infrastructure investments  \cite{Chiara,EricTim13}, as well as  reload and shout options \cite{kowk_reload_shout08}.  Since some of these applications also involve a sequence of perpetual call options, our analysis  is directly applicable and provides  an extension to discounting with  a negative rate.

Let us provide an outline of the paper. In Section \ref{sec:problem}, we formulate the optimal multiple stopping problem and   present some general  mathematical properties. In Section \ref{sec:single},  we   analyze both the single and multiple stopping problems driven by a two-sided L\'evy process. Section  \ref{section_numerical} discusses the numerical implementation and provides some illustrative numerical  examples. Section \ref{sect-conclude} concludes the paper. Our proofs, constituting a substantial part of the paper,  are included  in the Appendix.

\section{Problem Formulation and General Properties}\label{sec:problem}
In the background, we fix a probability space  $(\Omega, \F, \p)$  hosting a  \lev process $X=(X_t)_{t\ge 0}$ characterized uniquely by its \emph{Laplace exponent}
\begin{align}
\psi(\beta) := \log\E \left[ \mathrm{e}^{\beta X_1}\right] = c \beta
+\frac{1}{2}\sigma^2 \beta^2 +\int_{(-\infty,\infty)}(\mathrm{e}^{\beta
z}-1-\beta z 1_{\{|z|<1\}})\,\Pi(\diff z),  \label{laplace_exp}
\end{align}
for every $\beta\in\mathbb{C}$ such that $0\le\Re \beta<\beta_0$ (with $\Re z$ the real part of $z \in \mathbb{C}$)
where
\begin{align}
\beta_0:=\sup\{\beta\in\mathbb{R}\,:\,\ex [\mathrm{e}^{\beta X_1}]<\infty\}, \label{def_beta_0}
\end{align}
for some
$c \in\R$, $\sigma\geq 0$, 
 and a measure $\Pi$ with its support
$\R\backslash \{0\}$ such that
\begin{align}
\int_{(-\infty,\infty)} (1  \wedge z^2) \Pi(
\diff z)<\infty. \label{integrability_levy_measure}
\end{align}
We comment that the Laplace exponent $\psi$ can be extended beyond the line $\Re\beta=\beta_0$ by analytical continuation. Throughout,  we assume that $\beta_0>1$, and $-X$ is not a subordinator.

We denote  $\p_x$ as the  probability  and   $\E_x$ as the expectation with initial value $X_0 = x \in \R$.  When $X_0=0$, we drop the subscripts in $\p_x$ and $\E_x$.
Let $\mathbb{F}:=(\F_t)_{t\ge0}$ be the natural filtration generated by $X$.  The underlying price process is modeled by an exponential L\'evy process $S_t := \mathrm{e}^{X_t}$, $t\ge 0$.

Now we  describe our   optimal multiple stopping problem with a   refraction period  between  consecutive exercises. In the general setting, we can take the refraction period    $\delta$  as a deterministic constant or a positive random variable. We will  assume throughout the paper that the distribution of the random variable $X_\delta$ has no atoms.

Denote  by $\setT$  the set of $\mathbb{F}$-stopping times.  However, the incorporation of random refraction times requires us to expand the filtration.  For any collection $\Xi$ of positive random variables, we denote  $\mathbb{F}(\Xi)$ to be the smallest filtration such that all members of $\Xi$ are stopping times (see \cite{CIJ2013}).  
For each fixed $n\ge 1$, we introduce the set of admissible sequence of  exercise times:
\begin{multline*}
\setT^{(n)} :=\{\vec{\tau}=(\tau_n,\cdots,\tau_1)\,:\,  \tau_n \in  \setT,  \tau_i \text{  is  an } ~\mathbb{F}(\{\tau_j+ \delta_j\}_{j=i+1}^{n})\text{ -- stopping time,}  \\ \text{and ~}  \tau_{i+1}+\delta_{i+1} \le \tau_i, i = n-1, \cdots, 1\},
\end{multline*}
 where $\delta_i$'s are i.i.d.\ copies of some positive-valued random variable $\delta$, which are independent of the L\'evy process $X$.  The stopping time $\tau_i$ is an admissible exercise time when there are $i$ exercise opportunities left. In particular,  $\tau_n$ is the first exercise time and $\tau_1$ is the last one.

Throughout we  will work with the  reward  function \[\phi(x):=(\mathrm{e}^x-K)^+, \quad \forall x \in \R,  \] where  we call $K>0$ the strike price. With  $n\ge 1$ exercise opportunities, the  optimal    stopping problem is defined as \be
\tilde{v}^{(n)}(x) :=\sup_{\vec{\tau}\in \setT^{(n)}}\ex_x\left[\sum_{i=1}^n\mathrm{e}^{-\alpha\tau_i}\phi(X_{\tau_i})\ind_{\{\tau_i<\infty\}}\right], \quad \forall x \in \R.\label{Problem}
\ee

We impose a standing technical integrability condition to ensure that the problem is well defined.  
  \begin{assump}\label{assume1}There exists a constant  $\varrho>1$, such that the L\'evy process $X$ satisfies
\be\label{eq:ass}
\ex_x\!\left[\,\left(\sup_{0\le t<\infty} \mathrm{e}^{-\alpha t}\phi(X_t)\right)^\varrho\,\right]<\infty,\,\,\,\forall x\in\R.
\ee
\end{assump}
In Section 3, we will provide the conditions on $\alpha$ in  Assumption 3.1 so that this integrability condition will hold.

One key feature of our model is that the constant parameter $\alpha$ can be taken to be positive/negative, representing a discounting/inflating factor.  In the stock loan model proposed by Xia and Zhou \cite{xiazhou}, the negative effective discount rate  arises  when the rate charged by the bank $\gamma$  is  higher than the interest rate $r$. To see this, we consider  an investor who borrows amount  $K$ from a bank,  using  a share of stock $S$ as collateral. The borrower has the right to  redeem  the stock by paying the accrued principle $Ke^{\gamma t}$ at any time $t\ge 0$.  Hence, we  write the expected discounted payoff as $\ex_x [ \mathrm{e}^{-r\tau} (S_\tau - \mathrm{e}^{\gamma \tau}K)^+] = \ex_x [ \mathrm{e}^{-(r-\gamma)\tau} (\tilde{S}_\tau - K)^+]$ for $\tau\in\setT$, where $\alpha = r- \gamma<0$ and $\tilde{S}_\tau = \mathrm{e}^{-\gamma \tau} S_\tau$.   In a different class of  applications, the negative effective discount rate is also relevant  to real option exercise timing  when the investment cost $K$ grows at a rate $\gamma$ that is  higher than the firm's discount  rate $r$.

 In order to solve the   optimal multiple stopping problem \eqref{Problem}, we will establish its equivalence to the  following recursion of optimal single stopping problems:
\be
v^{(k)}(x):=\sup_{ {\tau}\in \setT}\ex_x\!\left[\mathrm{e}^{-\alpha\tau}\phi^{(k)}(X_{\tau})\ind_{\{\tau<\infty\}}\right]\label{single1}
\ee
where \be \phi^{(k)}(x) :=\phi(x)+\ex_x\!\left[\,\mathrm{e}^{-\alpha\delta}v^{(k-1)}(X_\delta)\,\right],\quad k = 1, 2, \cdots, n,\label{single2}\ee and  $v^{(0)}(x) := 0.$  To this end, we first present some useful properties of the value function $v^{(k)}$ for every $k\in\setN$.

\begin{lem}\label{lem_a}
For every integer $k\in \{1, \cdots, n\}$ and all $s\in\R_+:=(0,\infty)$, the function  $U^{(k)}(s):=v^{(k)}(\log s)$  is non-decreasing, convex, and hence differentiable almost everywhere on $\R_+$.
\end{lem}

As a result of the convexity and monotonicity, we obtain the existence and uniqueness of the  point of continuous fit for  $k=1$.
More specifically, using arguments as in the proof of Corollary 3.1 of \cite{xiazhou}, we know that there exists a level  $x_1^\star\in(\log K, \infty]$ such that $v^{(1)}(x)=\phi(x)>0$ if and only if $x\ge x_1^\star$. Note that we can without loss of generality rule out the possibility of  $x_1^\star\le \log K$ (exercising out of the money). 

For any $b\in\R$, we denote by $\tau_b^+$ the first up-crossing time
\[\tau^+_{b} = \inf\{\,t\ge 0 \,:\, X_t \ge b\}.\]
Here and throughout the paper we define $\inf \emptyset := \infty$. Furthermore,
for every $k\in\setN$, we define the value of discounted payoff of a threshold strategy $\tau_b^+\in\setT$ as
\be
\label{eq:gvalfunction}
g^{(k)}(x,b):=\ex_x\left[\mte^{-\alpha\tau_b^+}\phi^{(k)}(X_{\tau_b^+})\ind_{\{\tau_b^+<\infty\}}\right],\,\,\,\forall x\in\R.
\ee
When $x_1^\star < \infty$,  we know that the value function of the auxiliary problem \eqref{single1} for $k=1$ is given by
\[
v^{(1)}(x)=g^{(1)}(x,x_1^\star), \quad \forall x \in \R.
\]
When $x_1^\star = \infty$, the problem is trivial and the value function $v^{(1)}(x)$ is approximated by the expected value under $\tau_M^+$ by taking $M$ arbitrarily large.  We shall now assume the former
and give sufficient conditions so that similar results hold for the problem \eqref{single1} for $k\in\{2,\cdots, n\}$. To this end, we  adapt  the arguments from  the proof of  Lemma 3.2 of \cite{Carmona2008} to obtain the following result.

\begin{lem}\label{lem:x_k<x_1}
Suppose $x_1^\star \in(\log K, \infty)$.  Then
for  every  $1 \leq k \leq n$ and all $x\in[x_1^\star,\infty)$, we have $\phi^{(k)}(x)=v^{(k)}(x)$.
\end{lem}

Lemma \ref{lem:x_k<x_1} implies that $U^{(k)}(s)$ will eventually continuously fit $\phi^{(k)}(\log s)$  as    $s$ increases. By  the convexity of $U^{(k)}(s) \equiv v^{(k)}(\log s)$ and $\phi^{(k)}(\log s)$, we know that ${U^{(k)}}'(s)$ is almost everywhere bounded from above by $\mathrm{esssup}_{s\in\R_+}\left\{\frac{\partial}{\partial s}\phi^{(k)}(\log s)\right\}$.  In turn, we  deduce that $\ex[\mathrm{e}^{-\alpha\delta}v^{(k)}(x+X_\delta)]$ is differentiable in $x$ on $\R$ since  the distribution of $X_\delta$ does not charge a positive measure on the (at most) countable points where $v^{(k)}$ is not differentiable.
 
\begin{cor}\label{diff}
Suppose $x_1^\star \in(\log K, \infty)$ {and $\ex[\mte^{-\alpha\delta+X_\delta}]\le 1$}. For every integer $k \in\setN$, we have $0\le {v^{(k)}}'(x)\le k\mathrm{e}^x$, a.e. and
\[0\le \frac{\partial}{\partial x}\ex[\mathrm{e}^{-\alpha\delta}v^{(k)}(x+X_\delta)]=\ex[\mathrm{e}^{-\alpha\delta}{v^{(k)}}'(x+X_\delta)]\le k\mathrm{e}^{x},\,\,\,\forall x\in\R.\]
\end{cor}

We now  establish the equivalence between \eqref{Problem} and \eqref{single1}. Let us first recursively define  the set of stopping times
\bq \label{taus1}
\tau_n^\star&:=&\inf\{t\ge0\,:\, v^{(n)}(X_t)=\phi^{(n)}(X_t)\},\\
\tau_i^\star&:=&\inf\{t\ge\delta_{i+1}+\tau^\star_{i+1}\,:\,v^{(i)}(X_t)=\phi^{(i)}(X_t)\}, \quad \text{ for } i =n-1, \cdots, 1.\label{tausi}
\eq
We show below that   $(\tau_n^\star,\cdots,\tau_1^\star) \in\setT^{(n)}$ solve
the optimal multiple stopping problems  \eqref{Problem} and \eqref{single1}-\eqref{single2}.

\begin{thm}\label{prop:regularity} Suppose $x_1^\star \in(\log K, \infty)$. Fix a $k\in \{1, \cdots, n\}$, {then}  the stopping times $(\tau_i^\star)_{1\le i\le k}$ defined in \eqref{taus1}-\eqref{tausi} satisfy $\pr_x(\tau_i^\star<\infty,\,1\le i\le k)>0$, for all $x \in \R$.  Moreover,
\begin{enumerate}
\item[(i)] the value function of the auxiliary problem, $v^{(k)}$ of \eqref{single1}, satisfies
\bq
v^{(k)}(x)=\ex_x\left[\sum_{i=1}^k\mathrm{e}^{-\alpha\tau_i^\star}\phi(X_{\tau_i^\star})\ind_{\{\tau_i^\star<\infty\}}\right],\,\,\,\forall x\in\mathbb{R};\label{vn1}
\eq
\item[(ii)] the value function $v^{(k)}(x)$ of \eqref{single1} is equal to   $\tilde{v}^{(k)}(x)$ of \eqref{Problem} for every $x\in\R$;
\item[(iii)] for every initial value $X_0= x \in \R$,  all the  random variables in the collection
\[\S^{(k)}:=\{\mathrm{e}^{-\alpha \tau}v^{(k)}(X_\tau)\,:\,\tau\,\,\text{is an a.s. finite }\mathbb{F}\text{-stopping time}\}\] are uniformly bounded in $\mathrm{L}^{\varrho}(\diff\pr_x)$. 
\end{enumerate}
\end{thm}

We now return to the optimal multiple stopping problem \eqref{single1}-\eqref{single2}. From Lemma \ref{lem:x_k<x_1} we know that, if $x_1^\star\in(\log K, \infty)$, then for every $k\in\{2,\cdots,n\}$, we can define a finite level
\be\label{eq:def_x_k}
x_k^\star:=\inf\{x\le x_1^\star\,:\,v^{(k)}(y)=\phi^{(k)}(y), \,\,\,\forall y\ge x\}.
\ee
Then, for every $k\in\setN$,  the interval $[x_k^\star,\infty)$ must be  one connected domain of the optimal stopping region for problem \eqref{single1} with $k$ exercise opportunities. It should be noted  that, in general for  $k\ge2$, the optimal  stopping region can potentially be disconnected, consisting of  multiple disjoint intervals,  as the composite payoff function $\phi^{(k)}(\log s)$ is no longer  {piecewise} linear in $s\in\R_+$. However, if $[x_k^\star,\infty)$ is the only optimal stopping region, then the up-crossing time $\tau_{x_k^\star}^+$ must  be the optimal stopping time to problem \eqref{single1}, and $v^{(k)}(x)=g^{(k)}(x,x_k^\star)$ for all $x\in\R$.

To resolve  the issue of possible multiple disconnected components of optimal stopping for $k\ge2$, we consider the best threshold type strategy, among all first up-crossing times $\{\tau_b^+\,:\, b\in\R\}$,  and then give a sufficient condition for its optimality.
\begin{definition}\label{def_b_k}
We call a level $b_k^\star\in\R$ the optimal exercise threshold for problem \eqref{single1} with $k$ exercise opportunities, if and only if the function $g^{(k)}(x,b)$ is maximized at $b=b_k^\star$ for all $x\in\R$. More specifically, if $b_k^\star$ satisfies the following:
\begin{enumerate}
\item[(a)] For any fixed $x<b_k^\star$, the supremum of the function $g^{(k)}(x,\cdot)$ is given by $g^{(k)}(x,b_k^\star)$;
\item[(b)] For any fixed $x\ge b_k^\star$, the supremum of the function $g^{(k)}(x,\cdot)$ is given by $\phi^{(k)}(x)$.
\end{enumerate}
\end{definition}

When $k=1$, we know that $b_1^\star=x_1^\star$ if the latter is finite. Notice that,  for a general $k\ge2$, the optimal  exercise threshold $b_k^\star$ may not exist.  The following result characterizes the relationship between $x_k^\star$ and $b_k^\star$, when the latter exists.

\begin{prop}\label{x_k}
Suppose $x_1^\star \in(\log K, \infty)$.  Fix an integer  $k\in\{2,\cdots,n\}$, assume that $v^{(k-1)}(x)>v^{(k-2)}(x)$ for all $x\in\R$, and that $[x_{k-1}^\star,\infty)$ is the only optimal stopping region for problem \eqref{single1} with $(k-1)$ exercise opportunities. Then we have
\begin{enumerate}
\item[(i)] $v^{(k)}(x)>v^{(k-1)}(x)$ for all $x\in\R$;
\item[(ii)] $x_k^\star\in(\log K, x_1^\star]$, and if $b_k^\star$ exists, we also have $b_k^\star>\log K$;
\item[(iii)] If $b_k^\star$ exists and the process $(\mte^{-\alpha t}g^{(k)}(x,b_k^\star))_{t\ge0}$ is a $(\pr_x,\mathbb{F})$-supermartingale, then $x_k^\star=b_k^\star$ and $[x_k^\star,\infty)$ is the only optimal stopping region. Hence, the up-crossing time $\tau_{x_k^\star}^+$ is optimal.
   \end{enumerate}
\end{prop}

\begin{remark}\label{rmk:plus}
Proposition \ref{x_k} implies that each  value function can  be determined by first optimizing the expected reward over all candidate thresholds  \emph{that are above $\log K$}, followed by verifying the supermartingale property. Consequently, as far as the optimal thresholds are concerned, we can effectively remove the $+$ sign in the payoff function $\phi(x)$. From  Proposition  \ref{x_k} we conclude that $[x_k^\star,\infty)$ is the only optimal stopping region, and each  optimal exercise threshold $b_k^\star=x^\star_k$ exists and is bounded above by $x^\star_1$. 
\end{remark}

We now show that if $[x_{k}^\star,\infty)$ is the only connected optimal stopping region  for all $1\le k\le l$ for some $l\in\setN$, then  $(x_k^\star)_{1\le k\le l}$ is
non-increasing in $k$. To show this, we first prove that the process
\[V^{(k-1)}_t:=\mathrm{e}^{-\alpha t}\big(v^{(k-1)}(X_t)-v^{(k-2)}(X_t)\big), \quad  t \ge 0,\] is a supermartingale for any fixed $k\in\{2,\cdots,l+1\}$.

\begin{prop} \label{prop_mono} Suppose that $x_1^\star\in(\log K,\infty)$ and that $[x_{k}^\star,\infty)$ is the only connected optimal stopping region  for all $1\le k\le l$ for some $l\in\{2,\cdots,n-1\}$, then the sequence of optimal exercise thresholds  $(x_k^\star)_{1 \leq k \leq l+1}$ is  non-increasing in $k$, i.e.
\[\log K<x_{l+1}^\star\le x_{l}^\star\le \cdots\le x_1^\star,\]
and the process $(V^{(k-1)}_t)_{t \geq 0}$ is a $(\p_x, \mathbb{F})$-supermartingale of class ($D$) \citep[Chap. 3]{ProtterBook} for any $2\le k\le l+1$.
\end{prop}

Proposition \ref{prop_mono} tells us, if threshold type strategies are optimal for problem \eqref{single1} with $k$ exercise opportunities for all $1\le k\le l$, then the optimal exercise thresholds $(x_k^\star)_{1\le k\le l}$ are non-increasing in $k$. Hence, even if threshold type strategies are not optimal for problem \eqref{single1} with $l+1$ exercise opportunities, the optimal stopping region should contain $[x_l^\star,\infty)$. Moreover, for any number of remaining exercise opportunities, it is always optimal to exercise above  the strike price.

\section{Analytical Results}\label{sec:single}
In this section, we assume that $X$ is either a spectrally negative \lev process that is not the negative of a subordinator, or a \lev process with  an arbitrary negative jump distribution and a positive phase-type  jump distribution \cite{Asmussen2004}:
\begin{align} \label{XX}
  X_t  - X_0= J_t + \sum_{n=1}^{N_t} Z_n, \quad 0\le t <\infty.
\end{align}
Here, $(J_t)_{t\ge 0}$ is a spectrally negative \lev process  with or without a Brownian motion component, $(N_t)_{t\ge 0}$ is a Poisson process with arrival rate $\rho$, and  $Z =  ( Z_n)_{n = 1,2,\cdots }$ is an i.i.d.\ sequence of phase-type-distributed random variables with representation $(d,{\bm \alpha},{\bm T})$. In addition,  $J$ and $Z$ are mutually independent. For a comprehensive study on this process and its applications in American and Russian options, we refer the reader to  \cite{Asmussen2004}.

Recall that a  distribution on $\R_+$ is of phase-type if it is the distribution of the absorption time  in a finite state continuous-time Markov chain consisting of one absorbing state and $d\in\mathbb{N}$ transient states. Thus,  any phase-type distribution can be represented by $d$, the $d\times d$ transition intensity matrix over all  transient states $\bm{T}$, and the initial distribution of the Markov chain $\bm{\alpha}$. Without loss of generality, we assume that the positive phase-type  jump distribution is minimally represented with $d$ phases.
From \cite{Asmussen2004}, this guarantees that the singularities of the Laplace exponent $\psi$ with positive real part are eigenvalues of $\bm{T}$.  Moreover, by Theorem 5b on p.58 of \cite{Widder46}, we know that $\beta_0$ defined in \eqref{def_beta_0} is the smallest positive pole of $\psi$ and  $\lim_{\beta\uparrow\beta_0}\psi(\beta)=\infty$.  Henceforth, we  impose the following technical condition. 
\begin{assump}\label{assume2} The Laplace exponent $\psi$ and the discount rate $\alpha$ satisfy either (i) $\psi(1) <\alpha$,  or  (ii) $\psi(1) = \alpha<0$ and $\psi'(1)<0$.
\end{assump}
Under these conditions,  we shall show in the lemma below that the  optimal stopping problem in  \eqref{single1} is well-posed for each $1 \leq k \leq n$ in the sense that the integrability condition  in Assumption \ref{assume1} is met.  Also, the discounted price process $(\mathrm{e}^{-\alpha t+X_t})_{t\ge0}$ under $\pr$ has to be  a supermartingle, and not a martingale when $\alpha\ge 0$ \citep[Theorem 1]{mordecki2002}.  In effect,   the trivial optimal strategies of perpetual waiting are excluded.

\begin{lem}\label{lem:boundedness}
Assumption \ref{assume2} implies Assumption \ref{assume1}.
\end{lem}

We provide a detailed proof in Appendix \ref{Appen_1}. Next,  for an $\alpha\in\mathbb{R}$, we define $\Phi(\alpha)$ to be the largest  positive
root of $\psi(\beta)=\alpha$, which is a real number less than $\beta_0$, if it exists. Notice that  Assumption \ref{assume2} and $\beta_0>1$ imply that $\Phi(\alpha)$ exists and  $\Phi(\alpha)> 1$ and $\psi'(\Phi(\alpha))\ge0$.
 We denote the finite set of roots with positive real parts by

\be
{\mathcal{I}_\alpha=}
\{\,\rho_{i,\alpha}\,:\,\psi(\rho_{i,\alpha})=\alpha,\,\Re\rho_{i,\alpha}\ge\Phi(\alpha)\}_{1\le i\le |\mathcal{I}_\alpha|}, \label{I2}
\ee
where,  for the sake of mathematical convenience, we assume that the roots are all distinct for a given $\alpha\in\mathbb{R}$. It follows that $\rho_{1,\alpha}=\Phi(\alpha)<\beta_0$ and $\Re\rho_{i,\alpha}\ge\beta_0$ for all $i\ge 2$. Moreover, we remark  that $|\mathcal{I}_\alpha|=d$ or $d+1$ according to whether $-J$ is a subordinator or not, respectively \citep[Lemma 1]{Asmussen2004}. We label $\rho_{i,\alpha}$'s in such a way that $(\Re\rho_{i,\alpha},\,\Im\rho_{i,\alpha})$ is in ascending lexicographic order (here $\Im z$ is the imaginary part of any complex number $z$). Similarly, we define a second set of roots with positive real part, labeled in the same way as elements of $\mathcal{I}_\alpha$:
\be \label{J}\mathcal{J}:=\{\,\eta_j\,:\,\frac{1}{\psi(\eta_j)}=0,~ \Re\eta_j>0\,\}_{1\le j\le |\mathcal{J}|}, \ee
where multiple roots are counted  individually.
Notice that we have $\beta_0=\eta_1\in\mathcal{J}$ and $|\mathcal{J}|=d$.

\begin{remark} \label{remark_special_case_spec_negative}
If $X$ is a spectrally negative L\'evy process, i.e. $|\mathcal{J}|=0$, then, by our assumption that $-X$ is not a subordinator, we have  $\mathcal{I}_\alpha=\{\Phi(\alpha)\}$ and $|\mathcal{I}_\alpha|=1.$
\end{remark}

Fix $\alpha \geq 0$. Let  $\ep$ be an exponential random variable, with rate parameter $\alpha$, that is  independent of $X$. We follow the convention that $\mathbf{e}_0 =\infty$, $\pr$-a.s.  Then, it is known from Lemma 1 of \cite{Asmussen2004} that the Laplace transform of $\overline{X}_{\ep}:=\sup_{0 \leq t\le\ep}X_t$ is given by
\be\psi_\alpha^+(\beta):=
\ex [\mathrm{e}^{-\beta\overline{X}_{\ep}}]=\prod_{i=1}^{|\mathcal{I}_\alpha|}\frac{\rho_{i,\alpha}}{\rho_{i,\alpha}+\beta}\prod_{j=1}^{|\mathcal{J}|}\bigg(1+\frac{\beta}{\eta_j}\bigg)=\psi_\alpha^+(\infty)+\sum_{i=1}^{|\mathcal{I}_\alpha|}A_i\frac{\rho_{i,\alpha}}{\rho_{i,\alpha}+\beta},\,\,\,\forall \beta\ge 0,\label{Xmax_mgf}
\ee
with $\psi_\alpha^+(\infty)=\lim_{\beta\to\infty}\psi_\alpha^+(\beta)$, and the partial fraction coefficients:
\be \label{eq:Ai}A_i:=\prod_{\substack{j=1\\ j\neq i}}^{|\mathcal{I}_\alpha|}\frac{\rho_{j,\alpha}}{\rho_{j,\alpha}-\rho_{i,\alpha}}\prod_{j=1}^{|\mathcal{J}|}\frac{\eta_j-\rho_{i,\alpha}}{\eta_j}.\ee
As a result, the distribution of $\overline{X}_{\ep}$ is given by:
\be\label{Xmax_dist}
\pr(\overline{X}_{\ep}\in \diff x)=\sum_{i=1}^{|\mathcal{I}_\alpha|}A_i\rho_{i,\alpha}\mathrm{e}^{-\rho_{i,\alpha}x}\diff x,\,\,\,\forall x>0.
\ee

\begin{remark}As explained in \citep[Remark 4]{Asmussen2004}, the assumption of distinct roots is made for convenience. When there are multiple roots, the corresponding distribution $\pr(\overline{X}_{\ep}\in \diff x)$ will admit a form similar    to that  in \eqref{Xmax_dist} with different constant coefficients.   Moreover, the case with multiple roots only occurs for at most countably many values of $\alpha$ over $\mathbb{R}$. In other words, if one arbitrarily sets the values of $\gamma$ and $r$,  the probability of having multiple roots as a result is zero. 
\end{remark}

In the next subsection, we derive the value function and the optimal exercise threshold for the single stopping problem for any discount rate $\alpha$ satisfying Assumption \ref{assume2}.

\subsection{Optimal  Single  Stopping Problem}\label{subse:single}
If $\psi(1)<\alpha$ and $\alpha\ge 0$, then it is known from Theorem 1 of \cite{mordecki2002} that the optimal stopping time for an American call with strike price $K>0$ is given by
\be
\tau_1^\star=\inf\{\,t\ge0\,:\,X_t\ge\,\log (K\psi_\alpha^+(-1))\,\},\nn
\ee
and the value of the American call option is given by: \begin{align}
\ex_x\left[\mathrm{e}^{-\alpha\tau_1^\star}\phi(X_{\tau_1^\star})\ind_{\{\tau_1^\star<\infty\}}\right]
&=K\sum_{i=1}^{|\mathcal{I}_\alpha|} [K\psi_\alpha^+(-1)]^{-\rho_{i,\alpha}}\mathrm{e}^{\rho_{i,\alpha}x}\frac{A_{i}}{\rho_{i,\alpha}-1},\label{price}
\end{align}
where $A_{i}$'s are defined in \eqref{eq:Ai}.   The analogous expectation in  \eqref{price} for the case with  $\alpha<0$ can be computed   using the sets $\mathcal{I}_\alpha$ in \eqref{I2} and $\mathcal{J}$ in \eqref{J}, which we shall prove in Proposition \ref{value_at_passage} below.

In order  to address the case with a negative discount rate  ($\alpha<0$), one of our main steps is to apply a change of measures.   For $\kappa \in[ 0,\beta_0)$,
 we define  a new probability measure  $\p_x^\kappa\sim \p_x$ by
\begin{align}
\left. \frac {\diff \p_x^\kappa} {\diff \p_x}\right|_{\mathcal{F}_t} = \exp(\kappa (X_t-x) - \psi(\kappa) t), \quad t \geq 0. \label{change_of_measure}
\end{align}
Then, for $\beta > -\kappa$, the Laplace exponent of $X$  is given by  \citep[Theorem 3.9]{Kyprianou2006}
  \begin{align}
\psi_\kappa(\beta) :=\Big( \kappa \sigma^2 -c  + \int_{(-1,1)} z (\mathrm{e}^{\kappa z}-1) \Pi(\diff z)\Big) \beta
+\frac{1}{2}\sigma^2 \beta^2 +\int_{{(-\infty,\infty)}} (\mathrm{e}^{\beta z}-1 -\beta z 1_{\{ |z|<1\}}) \mathrm{e}^{\kappa z}\,\Pi(\diff z).\nn
\end{align}
Under the new probability measure $\p^\kappa$, the process is also a L\'evy process with a negative jump distribution and a positive phase-type  distribution, with a new scaled  \lev measure  $\Pi_\kappa (\diff u) := \mathrm{e}^{\kappa u}\,\Pi(\diff u)$.

\begin{prop}\label{lapfirstpass}  We extend the definition \eqref{Xmax_mgf} for $\alpha \leq 0$ and define the function $\psi_\alpha^+(\beta)$ and partial fraction coefficients $(A_i\equiv A(\rho_{i,\alpha}))_{1\le i\le|\mathcal{I}_\alpha|}$ using
\be\label{eq:psi-alpha}
\psi_\alpha^+(\beta):=\prod_{i=1}^{|\mathcal{I}_\alpha|}\frac{{\rho}_{i,\alpha}}{{\rho}_{i,\alpha}+\beta}\prod_{j=1}^{|\mathcal{J}|}\bigg(1+\frac{\beta}{{\eta}_j}\bigg) \equiv\psi_\alpha^+(\infty)+\sum_{i=1}^{|\mathcal{I}_\alpha|}A_{i}\frac{{\rho}_{i,\alpha}}{{\rho}_{i,\alpha}+\beta},\,\,\,\forall \beta\in\mathbb{C}.
\ee
Then for any fixed $b>x$ and $\beta\ge0$, we have
\be
\ex_x\left[\mte^{-\alpha\tau_b^+-\beta(X_{\tau_b^+}-b)}\ind_{\{\tau_b^+<\infty\}}\right]=\frac{1}{\psi_\alpha^+(\beta)}\sum_{i=1}^{|\mathcal{I}_\alpha|}A_i\frac{\rho_{i,\alpha}}{\rho_{i,\alpha}+\beta}\mte^{-\rho_{i,\alpha}(b-x)}.\label{exx}
\ee
\end{prop}

Using Proposition \ref{lapfirstpass}, we can compute the expected payoff of any threshold type strategy $\tau_b^+$. By Theorem 5b on p.58 of \cite{Widder46}, we can extend $\beta$  on  both sides of \eqref{exx} to a complex number as long as $\Re\beta>-\rho_{1,\alpha}$. In particular, by setting $\beta=-1,0>-\rho_{1,\alpha}$ in \eqref{exx}, we obtain the following result.

\begin{cor}\label{value_at_passage}
For all $b\ge \log K$ {and $b>x$} $($and hence $\phi(X_{\tau_b^+}) = \mathrm{e}^{X_{\tau_b^+}}-K$ on $\{\tau_b^+ < \infty \})$, we have
\be\label{first passage}
{g^{(1)}(x,b)}=\ex_x \left[\mathrm{e}^{-\alpha\tau_b^+}\phi(X_{\tau_b^+})\ind_{\{\tau_b^+<\infty\}}\right]=\frac{1}{\psi_\alpha^+(-1)}\cdot\sum_{i=1}^{|\mathcal{I}_\alpha|} A_{i}\mathrm{e}^{-{\rho_{i,\alpha}}(b-x)}\bigg(\frac{{\rho}_{i,\alpha}}{{\rho}_{i,\alpha}-1}\mathrm{e}^{b}-K\psi_\alpha^+(-1)\bigg).
\ee

\end{cor}

Since we already know the optimal stopping time is of threshold type when $k=1$, the analytic expression for the value function of the single stopping problem is then readily available to us by optimizing the exercise threshold $b$.
\begin{thm}\label{thm:value_single_stop}The optimal exercise threshold for the single stopping problem is given by $x_1^\star:=\log(K\psi_\alpha^+(-1))$ and that
the corresponding value function is given by
\be \label{valone}
v^{(1)}(x)=g^{(1)}(x,x_1^\star)=\left\{ \begin{array}{ll} \phi(x) , & x \geq x^\star_1, \\
K\cdot\sum_{i=1}^{|\mathcal{I}_\alpha|} [K\psi_\alpha^+(-1)]^{-\rho_{i,\alpha}}\mathrm{e}^{\rho_{i,\alpha}x}\frac{A_{i} }{\rho_{i,\alpha}-1}, & x < x^\star_1,
\end{array} \right.
\ee
where the function $g^{(1)}(\cdot,\cdot)$ is defined in \eqref{eq:gvalfunction}.
\end{thm}

\bs
\begin{remark}
Recently, Cai and Sun \cite{Cai2014} consider a single   stock loan problem under a hyper-exponential jump diffusion model, and  provide an analytic solution for the perpetual single stopping problem. In comparison, our Theorem 3.1 applies to  more general \lev models as described by  \eqref{XX} and Assumption \ref{assume2}.
\end{remark}
\begin{remark}\label{rmk:SNLP}
If  $X$ is a spectrally negative L\'evy process, then \eqref{valone} can be simplified to
\be
v^{(1)}(x)= \left\{ \begin{array}{ll} \phi(x) , & x \geq x^\star_1, \\ \phi(x^\star_1) \mathrm{e}^{-\Phi(\alpha)(x^\star_1-x)}, & x < x^\star_1,
\end{array} \right.
\ee
where $x^\star_1=\log(K\psi_\alpha^+(-1))$ with $\psi_\alpha^+(\beta)=\frac{\Phi(\alpha)}{\Phi(\alpha)+\beta}$. Notice that $\ex_x[\mathrm{e}^{-\alpha\tau_{b}^+}]=\mathrm{e}^{-\Phi(\alpha)(b-x)}$ for all  $x < b$.
\end{remark}

\subsection{Optimal Multiple Stopping Problem}
In this subsection, we characterize the optimal exercise thresholds that maximize $g^{(k)}(x,\cdot)$. First, 
recall from  Proposition \ref{lapfirstpass} that for all $x<b$, and the given $\alpha\in\mathbb{R}$  and $\beta\ge 0$,
\bq
\ex_x\left[\mathrm{e}^{-\alpha\tau_b^+-\beta (X_{\tau_b^+}-b)}\ind_{\{\tau_b^+<\infty\}}\right]&=&\frac{1}{\psi_\alpha^+(\beta)}\sum_{i=1}^{|\mathcal{I}_\alpha|}A_i\frac{\rho_{i,\alpha}}{\rho_{i,\alpha}+\beta}\mathrm{e}^{-\rho_{i,\alpha}(b-x)}.
\label{sumpipi}
\eq
The distribution of $X_{\tau_b^+}$ can be retrieved from \eqref{sumpipi} via inverse Laplace transform. To this end,  let us introduce
\bq
\phi_\infty&:=&\lim_{\beta\to\infty}{\beta\psi_\alpha^+(\beta)}=\left\{\begin{array}{ll}\frac{\prod_{i=1}^{|\mathcal{I}_\alpha|}\rho_{i,\alpha}}{\prod_{j=1}^{|\mathcal{J}|}\eta_j}>0,\,\,&\text{ {if $-J$ is not a subordinator}}\\
\infty, &\text{ {else}}\end{array}\right..\nn
\eq
Then there exists a unique (possibly signed) measure on $[0,\infty)$, $\nu(\diff y)$, such that 
\be\int_{[0,\infty)}\mte^{-\beta y}\nu(\diff y)=\frac{1}{\psi_\alpha^+(\beta)}-\frac{\beta}{\phi_\infty},\,\,\,\forall\beta\ge0.\label{eq:numeasure}\ee

\begin{remark}
If we assume that elements in $\mathcal{J}$ are all distinct,\footnote{This is the case, for example, when the upward jumps are hyper-exponential.} then we have,
\be
\nu(\diff y)=\left\{\begin{array}{ll}\frac{(\sum_{i=1}^{|\mathcal{I}_\alpha|}\rho_{i,\alpha}-\sum_{j=1}^{|\mathcal{J}|}\eta_j)}{\phi_\infty}\ind_{\{y=0\}}+\sum_{j=1}^{|\mathcal{J}|}\frac{1}{{\psi_\alpha^+}'(-\eta_j)}\mathrm{e}^{-\eta_jy}\diff y,~&\text{if $-J$ is not a subordinator}\\ \frac{1}{\psi_\alpha^+(\infty)}\ind_{\{y=0\}}+\sum_{j=1}^{|\mathcal{J}|}\frac{1}{{\psi_\alpha^+}'(-\eta_j)}\mathrm{e}^{-\eta_jy}\diff y, &\text{else}\end{array}\right.,\,\,\,\forall y\ge0.\nn
\ee 
\end{remark}

Furthermore, we define a measure on $[0,\infty)$ for each $1\le i\le |\mathcal{I}_\alpha|$:
\be\label{nudensity}
\bar{\nu}_i(\diff y):=\frac{\rho_{i,\alpha}}{\phi_\infty}\ind_{\{y=0\}}+\rho_{i,\alpha}\mte^{-\rho_{i,\alpha}y}\bigg(-\frac{\rho_{i,\alpha}}{\phi_\infty}+\int_{[0,y)}\mte^{\rho_{i,\alpha}z}\nu(\diff z)\bigg)\diff y,\,\,\,\forall y\ge0.
\ee
Then it can be easily verified that 
\be
\int_{[0,\infty)}\mte^{-\beta y}\bar{\nu}_i(\diff y)=\frac{\rho_{i,\alpha}}{\rho_{i,\alpha}+\beta}\frac{1}{\psi_\alpha^+(\beta)},\,\,\,\forall\beta\ge0, 1\le i\le |\mathcal{I}_\alpha|.\label{nubar}
\ee
As a result of \eqref{sumpipi} and \eqref{nubar}, we have for all $x<b$,
\bq
\label{eq:touch}\ex_x\left[\mte^{-\alpha\tau_b^+}\ind_{\{X_{\tau_b^+}=b\}}\ind_{\{\tau_b^+<\infty\}}\right]&=&\frac{1}{\phi_\infty}\sum_{i=1}^{|\mathcal{I}_\alpha|}A_i\rho_{i,\alpha}\mte^{-\rho_{i,\alpha}(b-x)},\\
\label{eq:overshoot}\ex_x\left[\mte^{-\alpha\tau_b^+}\ind_{\{X_{\tau_b^+}-b\in\diff y\}}\ind_{\{\tau_b^+<\infty\}}\right]&=&\sum_{i=1}^{|\mathcal{I}_\alpha|}A_i\mte^{-\rho_{i,\alpha}(b-x)}\bar{\nu}_i(\diff y),\,\,\,\forall y>0.
\eq

Equations \eqref{eq:touch} and \eqref{eq:overshoot} can be used to compute $\ex_x[\mathrm{e}^{-\alpha(\tau_b^++\delta)}v^{(k)}(X_{\tau_b^++\delta})]$. To this end, let $Y$ have the same distribution as $X_\delta$ under $\pr$, but is independent of $\mathcal{F}_{\tau_b^+}$. Then, for all $x<b$,
\begin{align}\label{delayed payoff}
&\ex_x\left[\mathrm{e}^{-\alpha(\tau_b^++\delta)}v^{(k)}(X_{\tau_b^++\delta})\ind_{\{\tau_b^+<\infty\}}\right]\nn\\
=&\sum_{i=1}^{|\mathcal{I}_\alpha|}A_i\mathrm{e}^{-\rho_{i,\alpha}(b-x)}\bigg(\int_{[0,\infty)}\ex[\mte^{-\alpha\delta}v^{(k)}(b+Y+y)]\bar{\nu}_i(\diff y)\bigg)\nn\\
=&\sum_{i=1}^{|\mathcal{I}_\alpha|}A_i\mathrm{e}^{-\rho_{i,\alpha}(b-x)}\bigg(\frac{\rho_{i,\alpha}}{\phi_\infty}\ex [\mathrm{e}^{-\alpha\delta}v^{(k)}(b+Y)]+\int_{(b,\infty)} \ex [\mathrm{e}^{-\alpha\delta}v^{(k)}(u+Y)]\bar{\nu}_i(-b+\diff u)\bigg).
\end{align}

Recall from Corollary \ref{diff} that $\ex[\mte^{-\alpha\delta}v^{(k)}(b+Y)]$ is  differentiable for all $b\in\R$.
As a result, the optimal exercise threshold $b_k^\star$ can be characterized by the first order condition: $\frac{\partial}{\partial b}|_{b=b_k^\star}g^{(k)}(x,b)=0$ for  any $x<b_k^\star$. In the remaining of this subsection, we will inductively prove that, if the threshold strategy is optimal for problem \eqref{single1} with up to $k-1$ exercise opportunities, for some $k\in\{2,\cdots,n\}$, then there exists a unique optimal exercise threshold $b_{k}^\star$ for problem \eqref{single1} with $k$ exercise opportunities.   To show   that the threshold type strategy $\tau_{b_k^\star}^+$ is indeed optimal over all $\mathbb{F}$-stopping times in $\setT$, we further prove that
the process $(\mte^{-\alpha t}g^{(k)}(X_t,b_k^\star))_{t\ge0}$ is a supermartingale.  Finally, based on the result for $k=1$ in Section \ref{subse:single}, mathematical induction will subsequently conclude the existence and uniqueness of the optimal exercise threshold $b_k^\star$, and the optimality of the threshold type strategy $\tau_{b_k^\star}^+$, for all $k\in\{2,\cdots,n\}$.

To facilitate later calculations, we define for all $1\le k\le n$ that
\be\label{function:u}
u^{(k)}(x):=\frac{{v_+^{(k)}}'(x)}{\phi_\infty}-\int_{[0,\infty)} v^{(k)}(x+y)\nu(\diff y),\ee 
 where ${v_+^{(k)}}'$ is the right derivative of $v^{(k)}$.
In particular, when $k=1$, we can use \eqref{valone} and \eqref{function:u} (see also the proof of Proposition \ref{prop:ukk} below using $v^{(1)}(x)=g^{(1)}(x,x_1^\star)$ and $v^{(0)}(x)\equiv u^{(0)}(x)\equiv0$) to obtain: 
\be
\label{eq:u1}
u^{(1)}(x)=\frac{\mte^{x_1^\star}-\mte^{x}}{\psival}\ind_{\{x\ge x_1^\star\}}.
\ee
Notice that $u^{(1)}$ is  continuous, non-positive  and non-increasing on $\R$ and is strictly decreasing on $[x_1^\star,\infty)$.

The following result characterizes the first order condition for the optimal exercise threshold $b_k^\star$.
\begin{prop}\label{thm:first order}
If for some fixed $k\in\{2,\cdots, n\}$, the threshold type strategy $\tau_{x_{k-1}^\star}^+$ is optimal for \eqref{single1} with $(k-1)$ exercise opportunities, $v^{(k-1)}(x)>v^{(k-2)}(x)$ for all $x\in\R$, and the function $u^{(k-1)}(x)$ is continuous on $\R$, then there exists at least one solution $b_k>\log K$ to the   equation:
\be\label{first order condition}
\tilde{u}^{k}_{0} (b_k) = 0, \quad \text{ where} ~~~\tilde{u}_0^{(k)}(x):=\frac{\mathrm{e}^{x_1^\star}-\mathrm{e}^{x}}{\psi_\alpha^+(-1)}+\ex[\mathrm{e}^{-\alpha\delta}u^{(k-1)}(x+X_\delta)],\,\,\,\forall x\in\R. \ee
Moreover, if the optimal exercise threshold $b_k^\star$ exists, it satisfies \eqref{first order condition}.
\end{prop}

In the next results, we establish the monotonicity of the  function $\tilde{u}_0^{(k)}$,  and the uniqueness of the solution to \eqref{first order condition}.

\begin{prop}\label{prop:ukk}
Under the assumption of Proposition \ref{thm:first order}, and further assume that
the function $u^{(k-1)}(x)$ is a non-increasing function. Let $b_k$ be any solution to \eqref{first order condition},  and define  
\be
\tilde{u}^{(k)}(x):=\frac{1}{\phi_\infty}\frac{\partial^+}{\partial x} {{g^{(k)}}(x,b_k)} -\int_{[0,\infty)} g^{(k)}(x+y, b_k)\nu(\diff y),\label{tildeuk}
\ee
where $\frac{\partial^+}{\partial x}$ is the right partial derivative operator.
Then,  $\tilde{u}^{(k)}$ is  continuous, non-positive, and  non-increasing on $\R$. Moreover, it can be expressed as 
\bq \label{tildeu2}
\tilde{u}^{(k)}(x)=\ind_{\{x\ge {b_k}\}}\tilde{u}_0^{(k)}(x).
\eq
\end{prop}
 We now prove the existence and the uniqueness of the optimal exercise threshold $b_k^\star$. 
\begin{lem}\label{uniqueness}
Under the condition of Proposition \ref{prop:ukk}, there is a unique solution to \eqref{first order condition}, and this solution is the optimal exercise threshold $b_k^\star$.
\end{lem}

 Hence,  Proposition \ref{thm:first order} applies, and  the optimal exercise threshold $b_k^\star$ is uniquely determined by \eqref{first order condition}. 

Next,  we prove the supermartingale property of the process $(\mte^{-\alpha t}g^{(k)}(X_t,b_k^\star))_{t\ge0}$, in order to show that the optimal stopping region for problem \eqref{single1} is one-sided, and hence, $b_k^\star=x_k^\star$ yields the optimal stopping time $\tau_{x_k^\star}^+.$ The main tool is to re-express the value of the threshold type strategy $g^{(k)}(x,b_k^\star)$, by an expectation of functionals of $\xq$ (see \cite{alili-kyp} for a similar solution approach).

 We begin with the case $k=1$. It can be easily seen from the proof of Proposition \ref{lapfirstpass} that
\be
v^{(1)}(x)=g^{(1)}(x,x_1^\star)=\lim_{q\downarrow0}\frac{\exph[\mte^{-\Phi(\alpha)\xq}(-u^{(1)}(x+\xq))]}{\exph[\mte^{-\Phi(\alpha)\xq}]}.\label{eq:v11}
\ee
We use \eqref{eq:v11} to initialize our induction step.

\begin{prop}\label{prop:verify}
Under the condition of Proposition \ref{prop:ukk}, and further assume that
\be
v^{(k-1)}(x)=g^{(k-1)}(x,x_{k-1}^\star)=\lim_{q\downarrow0}\frac{\exph[\mte^{-\Phi(\alpha)\xq}(-u^{(k-1)}(x+\xq))]}{\exph[\mte^{-\Phi(\alpha)\xq}]}.\label{eq:vgE}\ee 
Then we have
\be
g^{(k)}(x,b_k^\star)=\lim_{q\downarrow0}\frac{\exph[\mte^{-\Phi(\alpha)\xq}(-\tilde{u}^{(k)}(x+\xq))]}{\exph[\mte^{-\Phi(\alpha)\xq}]},\label{vEXP}
\ee
where the function $\tilde{u}^{(k)}$ is defined in \eqref{tildeuk}.
Moreover, the process $(\mte^{-\alpha t}g^{(k)}(X_t,b_k^\star))_{t\ge0}$ is a supermartingale.
\end{prop}

In summary, we   apply mathematical induction to  Propositions \ref{x_k}, \ref{prop_mono}, \ref{thm:first order}, \ref{prop:ukk} and  \ref{prop:verify} and Lemma \ref{uniqueness} to obtain the following result.
\begin{thm}\label{main2}
For every $k\in\setN$, the optimal stopping problem \eqref{single1} is solved by the up-crossing time $\tau_{x_k^\star}^+$, where $x_k^\star$ is the unique solution to \eqref{first order condition}, satisfying
\[\log K<x_n^\star\le x_{n-1}^\star\le \cdots\le x_1^\star.\]
The value function is given by
\[v^{(k)}(x)=g^{(k)}(x,x_k^\star)\,,\]
which can be expressed as \eqref{vEXP} above.  Moreover, the value functions are ordered as follows: 
\[0<v^{(1)}(x)<v^{(2)}(x)<\cdots<v^{(n)}(x),\,\,\,\forall x\in\R.\]
\end{thm}

\section{Numerical Examples } \label{section_numerical}
In this section, we present numerical examples based on our analytical results. In particular, we illustrate  the sensitivity of the optimal thresholds   with respect to the distribution of refraction times.  The numerical implementation is generally  challenging. It involves the evaluation of the expectation $\E_x \left[ \mathrm{e}^{- \alpha \delta}v^{(k-1)} (X_\delta) \right]$ while the distribution of the random variable $X_\delta$ is most commonly not explicit or even unknown.     As is used in \cite{ZeghalSwing}, Monte Carlo simulation is the most straightforward approach. However, it is far from being practical unless $k$ is a very small number. For refracted  multiple stopping problems, one needs to know the entire expected future payoff functional to carry out the backward induction. The simulation approach would require the computation of  these expectations for arbitrarily large number of starting points for each step, which adds to the computational burden and limits its applicability.  In particular, the payoff function of our problem is unbounded (and increases exponentially); the truncation that would be needed under the simulation method will produce non-negligible errors that would further be amplified as $k$ increases. 

For our numerical examples, we assume that $\delta$ is Erlang distributed (i.e. a sum of i.i.d. exponential random variables), and numerically solve for the optimal exercise  thresholds using the methods described in our separate paper \cite{TimKazuHZ14}.  The approach utilizes the resolvent measure (or the distribution of $X$ at an independent exponential random time) and carries out repeatedly and analytically the  integrations with respect to this measure. The resulting value functions are shown to be in a piecewise analytic form.  The results are exact when the jump size distribution is phase-type, and can be used as an approximation to problem with other \lev jumps  thanks to the denseness of the phase-type \lev processes.   The approach   can also  be applied  to  the case with  constant refraction times via the technique of Canadization.  For detailed analysis on its computational performance, we refer the reader to \cite{TimKazuHZ14}.

In our numerical results, we consider from  \eqref{XX} a spectrally negative \lev process with i.i.d.\ exponential jumps:
\begin{equation*}
  X_t  - X_0= \widetilde{c} t+\sigma B_t - \sum_{n=1}^{N_t} Z_n, \quad 0\le t <\infty, 
\end{equation*}
for some $\widetilde{c} \in \R$ and $\sigma \geq 0$.  Here $B=(B_t)_{ t\ge 0}$ is a standard Brownian motion, $N=(N_t)_{t\ge 0}$ is a Poisson process with arrival rate $\rho$, and  $Z = (Z_n)_{n = 1,2,\ldots}$ is an i.i.d.\ sequence of exponential random variables with parameter $\lambda > 0$.  These processes are assumed to be mutually independent.
For our studies below, we set  $\sigma = 0.2$, $\rho = \lambda = 1$ and $K = 50$.  Also, we use  $\alpha = -0.02$ and $\tilde{c} = 0.36$ so that $\alpha  -\psi(1) = 0.1 > 0$, which guarantees that Assumption \ref{assume2} is satisfied.  We consider two types of refraction times: (1) exponential and (2) Erlang with shape parameter $2$. We compute the results for a range of the expected refraction times, denoted by  $\bar{\delta}:=\E \delta$.

In Figure \ref{sensitivity_threshold}, we plot the optimal exercise thresholds $x_k^\star$  for $k=1,\ldots, 5$ against different  means of the  refraction time $\bar{\delta} =  0.5, 1.0, 1.5, \ldots, 10$.  Consistent with Proposition \ref{prop_mono},  the thresholds   monotonically decrease as  $k$ increases. In particular, the highest threshold corresponds to the last remaining exercise ($k = 1$). In this case,   the refraction time is completely  irrelevant, so the threshold value stays constant over different mean refraction times under any distribution. Interestingly, with $k$ fixed, the thresholds are \emph{not} monotone in the mean refraction time. On one hand, refraction times are constraints on the stopping times, so they reduce  the value functions but not necessarily the exercise thresholds.      Intuitively, a very long refraction time reduces the value of subsequent exercise opportunities, and  incentivizes the holder to focus more on the next immediate stopping. This helps explain that the thresholds tend to be closer for very long mean  refraction times.

\begin{figure}[h!]
\begin{small}
\begin{center}
\begin{minipage}{1.0\textwidth}
\centering
\begin{tabular}{cc}
   \includegraphics[scale=0.5]{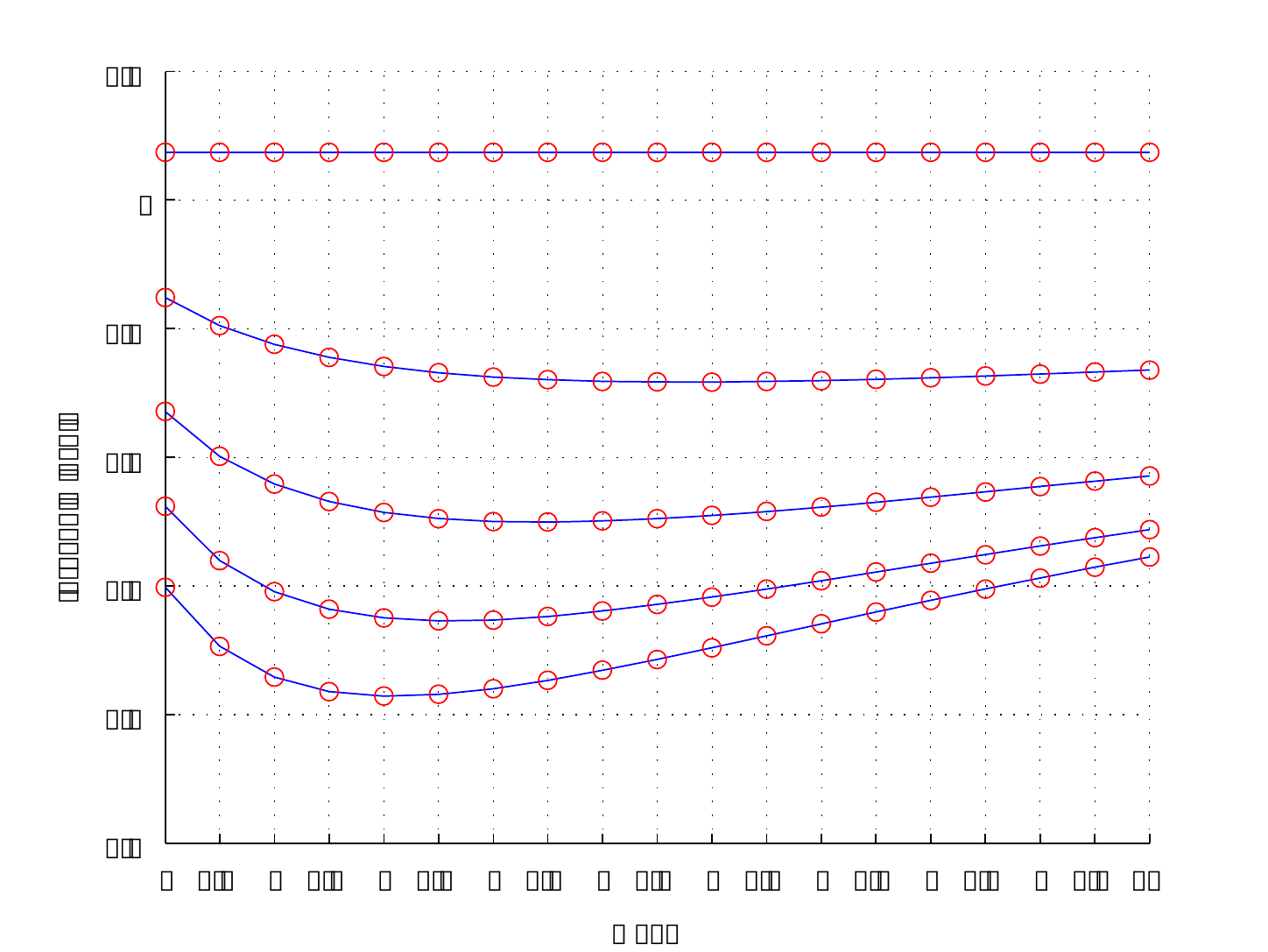} & \includegraphics[scale=0.5]{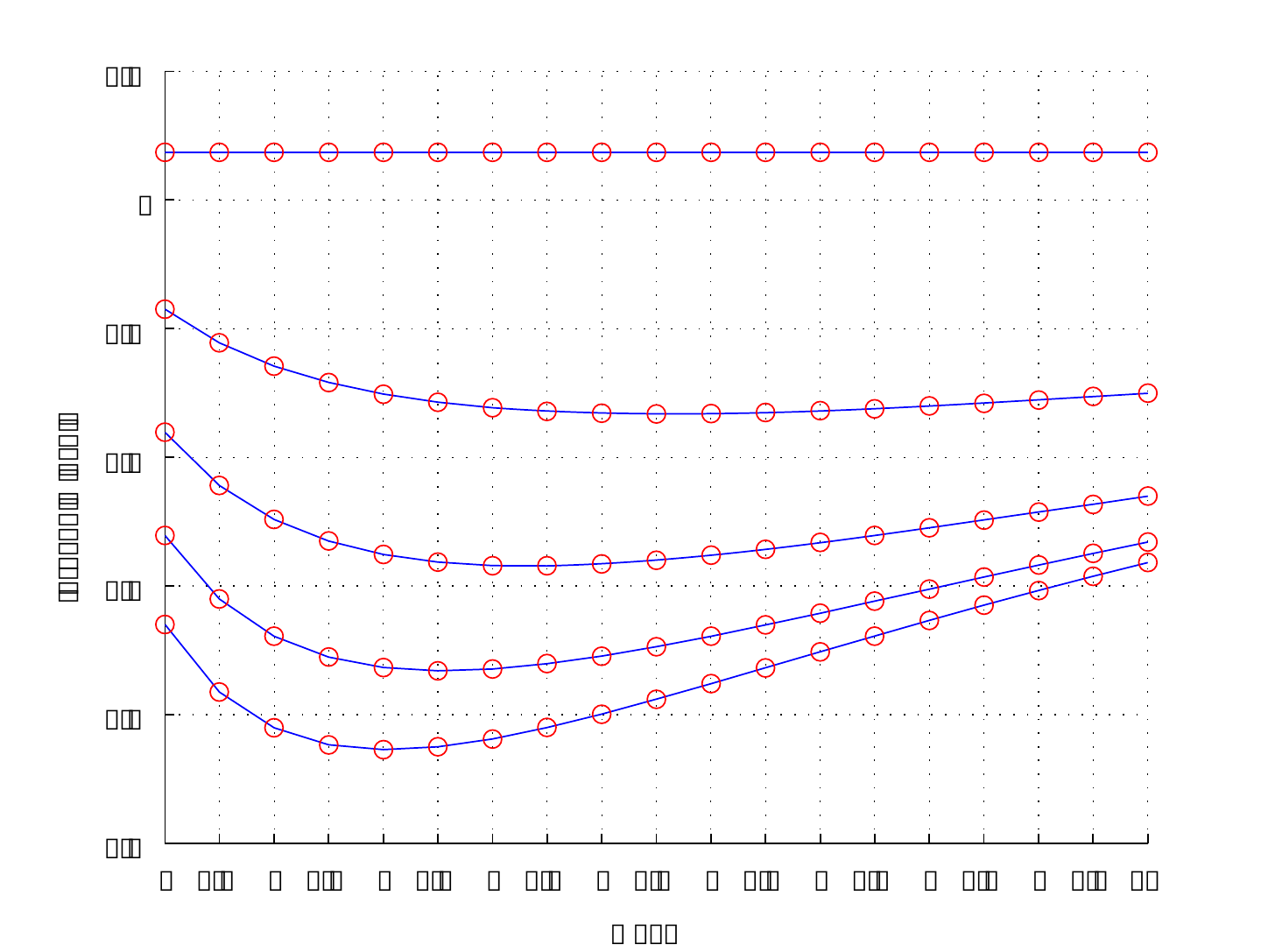}  \\
(1) Exponential & (2) Erlang \vspace{0.1cm} \\
\end{tabular}
\end{minipage}
\caption{Optimal multiple exercise thresholds $x_k^\star$ for $k=1,2,3,4,5$, plotted for different  mean refraction times $\bar{\delta}$ when they are exponentially  (left) and Erlang (right) distributed respectively.} \label{sensitivity_threshold}
\end{center}
\end{small}
\end{figure}

\section{Conclusions}\label{sect-conclude}
We have studied an optimal multiple stopping problem with the features of negative discount rate and random refraction times under a general class of \lev models.  In order to account for the negative discount rate,   the technique of change of measure is shown to be very useful though the analysis under the new measure is challenging.  As seen in Theorems  \ref{thm:value_single_stop} and \ref{main2} above, the   optimal exercise  thresholds are determined explicitly in the single stopping case and recursively in the multiple stopping case. While our problem setting is selected with the application to stock loans in mind, the current paper also presents a blueprint to rigorously analyze perpetual optimal refracted multiple stopping problems with alternative payoffs, such as put options. These would be natural directions for future research.

\appendix
\section{Appendix}
\subsection{Proof of Lemma \ref{lem_a}}\label{App:1}
We begin by noticing that, for any fixed $k\in\setN$,
\[U^{(k)}(s)=\sup_{\tau\in\setT}\ex \left[\mathrm{e}^{-\alpha\tau}\phi^{(k)}(\log s+X_\tau)\ind_{\{\tau<\infty\}}\right].\]
For $k=1$, for any stopping time $\tau$ and any $s_1,s_2\in\R_+$ such that $s_1>s_2$, it follows from the monotonicity of $\phi$ that
\bq
U^{(1)}(s_1)=\sup_{\tau\in\setT}\ex \left[\mathrm{e}^{-\alpha\tau}\phi(\log s_1+X_\tau)\ind_{\{\tau<\infty\}}\right]
&\ge&\sup_{\tau\in\setT}\ex \left[\mathrm{e}^{-\alpha\tau}\phi(\log s_2+X_\tau)\ind_{\{\tau<\infty\}}\right]=U^{(1)}(s_2).\nn
\eq
Similarly, from the subadditivity of supremum and the convexity of $(\phi\circ\log)$, we have,  for any $p,q>1$ such that $\frac{1}{p}+\frac{1}{q}=1$, and any $s_1, s_2\in\R_+$, that
\begin{multline}
\frac{U^{(1)}(s_1)}{p}+\frac{U^{(1)}(s_2)}{q}=\frac{1}{p}\sup_{\tau\in\setT}\ex\left[\mathrm{e}^{-\alpha\tau}\phi(\log(s_1\exp(X_\tau)))\ind_{\{\tau<\infty\}}\right]+\frac{1}{q}\sup_{\tau\in\setT}\ex\left[\mathrm{e}^{-\alpha\tau}\phi(\log(s_2\exp(X_\tau)))\ind_{\{\tau<\infty\}}\right]\\
\ge\sup_{\tau\in\setT}\ex\left[\mathrm{e}^{-\alpha\tau}\bigg(\frac{\phi(\log(s_1\exp(X_\tau)))}{p}+\frac{\phi(\log(s_2\exp(X_\tau)))}{q}\bigg)\ind_{\{\tau<\infty\}}\right]\\
\ge\sup_{\tau\in\setT}\ex\left[\mathrm{e}^{-\alpha\tau}(\phi\circ\log)\bigg(\Big(\frac{s_1}{p}+\frac{s_2}{q} \Big)\exp(X_\tau)\bigg)\ind_{\{\tau<\infty\}}\right]=\sup_{\tau\in\setT}\ex_{\log(\frac{s_1}{p}+\frac{s_2}{q})}\left[\mathrm{e}^{-\alpha\tau}\phi(X_\tau)\ind_{\{\tau<\infty\}}\right]=U^{(1)}\bigg(\frac{s_1}{p}+\frac{s_2}{q}\bigg).\nn
\end{multline}
Hence, the convexity holds also  for $k=1$. This  implies that  $U^{(1)}(s)$ is differentiable almost everywhere on $\R_+$.

Now suppose that the claim is true for $k=l-1$ for some $l\in\{2,\cdots, n-1\}$; that is,  $U^{(l-1)}$ is non-decreasing and convex. By the same argument  above, we conclude that $\ex[\mathrm{e}^{-\alpha\delta}U^{(l-1)}( s\exp(X_\delta))]$ is also non-decreasing and convex.
This implies that  $\phi^{(l)}$ and  $v^{(l)}$ all have the  monotonicity and convexity properties. By induction, we conclude.

\subsection{Proof of Corollary \ref{diff}}
For $k=1$ we have $0\le {U^{(1)}}'(s)= {v^{(1)}}'(\log s) /s\le 1$,  for a.e.\ $s\in\R_+$, or equivalently $0\le {v^{(1)}}'(x)\le \mathrm{e}^x$, a.e.\  $x\in\R$ and $|v^{(1)}(x+\epsilon)-v^{(1)}(x)|\le \mte^x|\mte^\epsilon-1|$,  for any $x,\epsilon\in\R$.
Let us denote by $D^{(1)}:=\{x\in\R\,:\, {U^{(1)}_+}'(x)>{U^{(1)}_-}'(x)\}=\{x\in\R\,:\, {v^{(1)}_+}'(x)>{v^{(1)}_-}'(x)\}$, where ${U_+^{(1)}}'$ and ${U_-^{(1)}}'$ (${v_+^{(1)}}'$ and ${v_-^{(1)}}'$, resp.) are, respectively, the right and left derivatives of $U^{(1)}$ ($v^{(1)}$, resp.). Then we know that $D^{(1)}$ is at most a countable set by Lemma \ref{lem_a}.
On the other hand, using the fact that $X_\delta$ has no atom, we know that, for any fixed $x\in\R$, we have for $\p$-almost every $\omega \in \Omega$,  $x+X_{\delta(\omega)}(\omega)\notin D^{(1)}$. It follows that $\mathrm{e}^{-\alpha\delta(\omega)}v^{(1)}(x+X_{\delta(\omega)}(\omega))$ is almost surely differentiable at this fixed $x\in\R$, and $0\le \mathrm{e}^{-\alpha\delta }{v^{(1)}}'(x+X_{\delta})\le \mathrm{e}^{-\alpha\delta +X_{\delta}}\cdot\mathrm{e}^x$, $\pr$-a.s. By the assumption, the nonnegative random variable $\mathrm{e}^{-\alpha\delta+X_{\delta}}\cdot\mathrm{e}^x$ has a finite expectation.
Now for any $\epsilon\neq0$, as $\epsilon\to 0$, we have by Jensen's inequality that
\begin{align*}
{0\le\,} &\bigg|\frac{\ex[\mte^{-\alpha\delta}(v^{(1)}(x+\epsilon+X_\delta)-v^{(1)}(x+X_\delta))]}{\epsilon}-\ex[\mte^{-\alpha\delta}{v^{(1)}}'(x+X_\delta)]\bigg|\\
{\le}\, &\ex\bigg[\mte^{-\alpha\delta}\bigg|\frac{v^{(1)}(x+\epsilon+X_\delta)-v^{(1)}(x+X_\delta)}{\epsilon}-{v^{(1)}}'(x+X_\delta)\bigg|\bigg]\to 0,\end{align*}  where we have  used the fact that $|\frac{v^{(1)}(x+\epsilon+X_\delta)-v^{(1)}(x+X_\delta)}{\epsilon}-{v^{(1)}}'(x+X_\delta)|\le \mte^{x+X_\delta}(\frac{|\mte^{\epsilon}-1|}{|\epsilon|}+1)$, and
\[\mte^{-\alpha\delta}\cdot\lim_{\epsilon\to0}\bigg|\frac{v^{(1)}(x+\epsilon+X_\delta)-v^{(1)}(x+X_\delta)}{\epsilon}-{v^{(1)}}'(x+X_\delta)\bigg|=0,\,\,\,\pr\text{-a.s.}\]
 and the dominated convergence theorem.
It follows that $\ex[\mathrm{e}^{-\alpha\delta}v^{(1)}(x+X_{\delta})]$ is differentiable in $x$, and
\[0\le\frac{\partial}{\partial x}\ex[\mathrm{e}^{-\alpha\delta}v^{(1)}(x+X_\delta)]=\ex[\mathrm{e}^{-\alpha\delta}{v^{(1)}}'(x+X_\delta)]\le\ex[\mathrm{e}^{-\alpha\delta+X_\delta}]\mathrm{e}^x\le\mathrm{e}^x.\]

Now suppose that the claim is true for $k=l-1$ for some $ l\in\{2,\cdots, n-1\}$. Then we know that $\phi^{(l)}(x)$ is differentiable on {$\R\backslash\{\log K\}$}  and its derivative admits the upper bound
\[{\phi^{(l)}}'(x)\, {\le}\, \mathrm{e}^x+\frac{\partial}{\partial x}\ex[\mathrm{e}^{-\alpha\delta}v^{(l-1)}(x+X_\delta)]\le \mathrm{e}^x+(l-1)\mathrm{e}^x=l\mathrm{e}^x,\,\,{\forall x\neq\log K}.\]
That is, $0\le \frac{\partial}{\partial s}\phi^{(l)}(\log s)\le l$ for all {$s\in \R_+\backslash\{K\}$}. It follows that $0\le {U^{(l)}}'(s)\le l$, a.e. which implies that $0\le {v^{(l)}}'(x)\le le^{x}$, a.e. By the same arguments as above, we also obtain that, for all $x\in\R$,
\[0\le \frac{\partial}{\partial x}\ex[\mathrm{e}^{-\alpha\delta}v^{(l)}(x+X_\delta)]\le l\mathrm{e}^{x}.\]
The result now follows from mathematical induction.

\subsection{Proof of Theorem \ref{prop:regularity}}
First,   Lemma   \ref{lem:x_k<x_1} implies  that $\pr_x(\tau_i^\star<\infty, 1\le i\le k)>0$  holds provided that $x^\star_1<\infty$ and $-X$ is not a subordinator. 

Next, following   Lemma 2.1 of  \cite{Carmona2008}, we deduce  recursively  that
\be v^{(k)}(x)\,\le\,\ex_x\left[\sup_{0\le t<\infty}\mathrm{e}^{-\alpha t}\phi^{(k)}(X_t)\right]\le \left(\ex_x\left[\left(\sup_{0\le t<\infty}\mathrm{e}^{-\alpha t}\phi^{(k)}(X_t)\right)^\varrho\right]\right)^{\frac{1}{\varrho}}<\infty,\label{eq:regularity}\ee
for all $k\ge1$. Hence, the single optimal stopping  problem \eqref{single1} is well defined. To ensure the existence of an optimal stopping time,  we may  adapt the proof of Proposition 3.2 in  \cite{ZeghalSwing} to the setting with possibly negative discount rate $\alpha$ and call-like payoff. More precisely, by Lemma \ref{lem_a} and Corollary \ref{diff} we know that $U^{(k)}(s)$ is globally Lipschitz in $s\in\R_+$, which implies that, by the proof of Proposition 3.2 in  \cite{ZeghalSwing},
the expected jump of $\mathrm{e}^{-\alpha \tau}v^{(k)}(X_\tau)$, at any predictable time $\tau$, is zero, namely,
\[\ex_x[\Delta(\mathrm{e}^{-\alpha \tau}v^{(k)}(X_\tau))\ind_{\{\tau<\infty\}}]=0, \quad k \ge 1,\]
where $\Delta(\mathrm{e}^{-\alpha \tau}v^{(k)}(X_\tau)):=\mathrm{e}^{-\alpha \tau}[v^{(k)}(X_\tau)-v^{(k)}(X_{\tau-})]$. This implies  that the Snell envelope $(\mathrm{e}^{-\alpha t}v^{(k)}(X_t))_{t\ge0}$ is left-continuous in expectation. In turn, this allows us to apply the arguments in Theorem 2.1 of \cite{Carmona2008} to conclude \eqref{vn1}.
This proves {\em(i)}.

To prove {\it(ii)}, we first comment that the result  holds trivially for the case $k=1$. For $k\in\{2,\cdots, n\}$, we observe from \eqref{vn1} that  $v^{(k)}(x) \le \tilde{v}^{(k)}(x)$ since $(\tau_i^\star)_{1\le i\le k}$ are admissible candidate stopping times (see\eqref{taus1}-\eqref{tausi}). 
The reverse inequality can be proved by induction. To this end,
notice that  $v^{(1)}(x) \ge \ex_x[ \mathrm{e}^{-\alpha \nu} \phi(X_\nu)\ind_{\{\nu<\infty\}}]$ for any arbitrary $\mathbb F$-stopping time $\nu\in\setT$ by \eqref{single1} for $k=1$.  Now by applying  \eqref{single1}, \eqref{single2} and repeated expectations, we get  
\begin{align} \label{v22}v^{(2)}(x)  &\ge\ex_x\left[\mathrm{e}^{-\alpha\tau}\left(\phi(X_\tau)+\ex_{X_\tau}\left[\mathrm{e}^{-\alpha\delta} v^{(1)}(X_\delta)\right]\right)\ind_{\{\tau<\infty\}}\right]\nn\\
&\ge\ex_x\left[\mathrm{e}^{-\alpha\tau}\left(\phi(X_\tau)+\ex_{X_\tau}\left[\mathrm{e}^{-\alpha\delta} \ex_{X_\delta}\left[\mathrm{e}^{-\alpha(\nu-\delta-\tau)}\phi(X_{\nu-\delta-\tau})\ind_{\{\nu<\infty\}}\right]\right]\right)\ind_{\{\tau<\infty\}}\right]\nn\\
&=\ex_x\!\left[\left(\mathrm{e}^{-\alpha\tau}\phi(X_{\tau})\ind_{\{\tau<\infty\}} + \mathrm{e}^{-\alpha\nu} \phi(X_\nu)\ind_{\{\nu<\infty\}} \right)\right] \end{align} for every  $\mathbb F$-stopping time $\tau$ and $\mathbb F(\tau+\delta)$-stopping time $\nu \ge \tau + \delta$. Maximizing \eqref{v22} over $(\tau,\nu)\in\setT^{(2)}$ yields that $v^{(2)}(x)\ge \tilde{v}^{(2)}(x)$.  The result now follows from mathematical induction.

Finally, for (iii), for any finite $\mathbb{F}$-stopping time $\tau\in\setT$, by the strong Markov property, we have
\bq&&\ex_x\left[\left(\mathrm{e}^{-\alpha \tau}v^{(k)}(X_\tau)\right)^\varrho\right]\le \ex_x\left[\left(\mathrm{e}^{-\alpha \tau}\cdot\ex_{X_\tau}\left[\sup_{0\le s<\infty}\mathrm{e}^{-\alpha s}\phi^{(k)}(X_s)\right]\right)^\varrho\right]\nn\\
&\le&\ex_x\left[\left(\mathrm{e}^{-\alpha \tau}\right)^{\varrho}\cdot\ex_{X_\tau}\left[\left(\sup_{0\le s<\infty}\mathrm{e}^{-\alpha s}\phi^{(k)}(X_{s})\right)^\varrho\right]\right]=\ex_x\left[\sup_{0\le s<\infty}\left(\mathrm{e}^{-\alpha(\tau+s)}\phi^{(k)}(X_{\tau+s})\right)^\varrho\right]\nn\\
&\le&\ex_x\left[\sup_{0\le s<\infty}\left(\mathrm{e}^{-\alpha s}\phi^{(k)}(X_s)\right)^\varrho\right]<\infty,\nn
\eq
where the first and second inequalities follow from \eqref{eq:regularity} and the equality is due to repeated expectations. Hence we have the uniform boundedness of elements of $\S^{(k)}$ in $\mathrm{L}^\varrho(\diff\pr_x)$. This proves {\em (iii)} and completes  the proof.

\subsection{Proof of Proposition \ref{x_k}}
 If $v^{(k-1)}(x)>v^{(k-2)}(x)$ for all $x\in\R$,
then by \eqref{single2}, we know that $\phi^{(k)}(x)>\phi^{(k-1)}(x)$ for all $x\in\R$.
Furthermore, if $[x_{k-1}^\star,\infty)$ is the only optimal stopping region for problem \eqref{single1} with $(k-1)$ exercise opportunities, then the up-crossing time $\tau_{x_{k-1}^\star}^+$ is the optimal stopping time. Hence, for all $x\in\R$, we prove  {\em(i)} through the inequality
\begin{align*}v^{(k)}(x)\ge g^{(k)}(x,x_{k-1}^\star)=&\ex_x\left[\mathrm{e}^{-\alpha\tau_{x_{k-1}^\star}^+}\phi^{(k)}(X_{\tau_{x_{k-1}^\star}^+})\ind_{\{\tau_{x_{k-1}^\star}^+<\infty\}}\right]\\
>&\ex_x\left[\mathrm{e}^{-\alpha\tau_{x_{k-1}^\star}^+}\phi^{(k-1)}(X_{\tau_{x_{k-1}^\star}^+})\ind_{\{\tau_{x_{k-1}^\star}^+<\infty\}}\right]=g^{(k-1)}(x,x_{k-1}^\star)=v^{(k-1)}(x).\end{align*}

To prove {\em (ii)},  we first recall from Lemma \ref{lem:x_k<x_1} that $x_k^\star \in (-\infty, x_1^\star]$.  Hence, for the first claim, it is sufficient to show that $v^{(k)}(x) \neq \phi^{(k)}(x)$ on $(-\infty, \log K]$.  Indeed,  we use
 the supermartingale property of value functions to obtain that $\ex_x[\mathrm{e}^{-\alpha\delta}v^{(k-1)}(X_\delta)]\le v^{(k-1)}(x)$ for all $x\in\R$. Therefore, for all $x\le\log K$,
 \[v^{(k)}(x)>v^{(k-1)}(x)\ge \ex_x[\mathrm{e}^{-\alpha\delta}v^{(k-1)}(X_\delta)]{=}\,\phi^{(k)}(x).\]
It follows that $x_k^\star\in(\log K, x_1^\star]$.
Similarly, if the optimal exercise threshold $b_k^\star$ exists, then we have
\[g^{(k)}(x,b_k^\star)\ge g^{(k)}(x,x_{k-1}^\star)>g^{(k-1)}(x,x_{k-1}^\star)=v^{(k-1)}(x)\ge\ex_x[\mte^{-\alpha\delta}v^{(k-1)}(x)].\]  
Therefore, for $x=\log K$, we have $g^{(k)}(\log K,b_k^\star)>\phi^{(k)}(\log K)$, which implies that $b_k^\star>\log K$.

We now proceed to prove {\em(iii)} by establishing the  sufficient conditions for optimality (see e.g. \citep[Sect. 6]{alili-kyp}). If the optimal exercise level  $b_k^\star$ exists,  then it is easily seen that
\begin{enumerate}
\item[(a)] for  all $x\in\R$, $g^{(k)}(x,b_k^\star)\ge g^{(k)}(x,x)=\phi^{(k)}(x)$, and $g^{(k)}(x,b_k^\star)>0$ by the fact that $b_k^\star>\log K$ (see {\em(ii)});
\item[(b)] for all  $x\in[b_k^\star,\infty)$, we have $g^{(k)}(x,b_k^\star)=\phi^{(k)}(x)$;
\item[(c)] for all $t>0$, by the strong Markov property of $X$, we have
\begin{align*}
g^{(k)}(x,b_k^\star)=&\ex_x\left[\mte^{-\alpha\tau_{b_k^\star}^+}\phi^{(k)}(X_{\tau_{b_k^\star}^+})\ind_{\{\tau_{b_k^\star}^+<\infty\}}\right]=\ex_x\bigg[\ex_x\bigg[\mte^{-\alpha\tau_{b_k^\star}^+}\phi^{(k)}(X_{\tau_{b_k^\star}^+})\ind_{\{\tau_{b_k^\star}^+<\infty\}}\bigg|\mathcal{F}_{t}\bigg]\bigg]\\
=&\ex_x\bigg[\mte^{-\alpha\tau_{b_k^\star}^+}\phi^{(k)}(X_{\tau_{b_k^\star}^+})\ind_{\{\tau_{b_k^\star}^+\le t\}}\bigg]+\ex_x\bigg[\ind_{\{\tau_{b_k^\star}^+>t\}}\ex_{X_t}\bigg[\mte^{-\alpha\tau_{b_k^\star}^+}\phi^{(k)}(X_{\tau_{b_k^\star}^+})\ind_{\{\tau_{b_k^\star}^+<\infty\}}\bigg]\bigg]\\
=&\ex_x\left[\mte^{-\alpha(t\wedge\tau_{b_k^\star}^+)}g^{(k)}(X_{t\wedge\tau_{b_k^\star}^+},b_k^\star)\right].
\end{align*}
That is, the stopped process $(\mte^{-\alpha(t\wedge\tau_{b_k^\star}^+)}g^{(k)}(X_{t\wedge\tau_{b_k^\star}^+},b_k^\star))_{t\ge0}$ is a $(\pr_x,\mathbb{F})$-martingale for any fixed $x\in\R$.
\end{enumerate}

Now if we know additionally that $(\mte^{-\alpha t}g^{(k)}(X_t,b_k^\star))_{t\ge0}$  is a supermartingale, then we can conclude that $v^{(k)}(x)=g^{(k)}(x,b_k^\star)$, and $[x_k^\star,\infty)=[b_k^\star,\infty)$ is the only stopping region.

\subsection{Proof of Proposition \ref{prop_mono}}
Clearly the claim holds for $l=2$   as we already know that $x_2^\star\le x_1^\star$, $V^{(1)}_t=\mathrm{e}^{-\alpha t}v^{(1)}(X_t)$ is a $(\pr_x,\mathbb{F})$-supermartingale,  and all random variables in the  set $\S^{(1)}$ in Theorem \ref{prop:regularity}  are uniformly bounded in $\mathrm{L}^\varrho(\diff\p_x)$.
If $l\ge 3$, suppose the claim is true  for $k=2, \cdots,h$ for some $h\in\{2,\cdots,l-1\}$. That is,
\[\log K<x_h^\star\le x_{h-1}^\star\le\cdots\le  x_1^\star\qquad\text{ and }\qquad \ex_x[V_t^{(k-1)}]\le V_0^{(k-1)}=v^{(k-1)}(x)-v^{(k-2)}(x),\,\,\forall k\in\{2,\cdots,h\}.\]  
Now, from the general theory of optimal stopping we know that the stopped process $(\mathrm{e}^{-\alpha(t\wedge\tau_{x_h^\star}^+)}v^{(h)}(X_{t\wedge\tau_{x_h^\star}^+}))_{t\ge0}$ is a martingale (see e.g. \citep[Sect. 6]{alili-kyp}). Therefore, for $x_h^\star=\min_{1 \leq k\le h}x_k^\star$, we know that the stopped process $(V_{t\wedge\tau_{x_h^\star}^+}^{(l)})_{t\ge0}$ is a martingale. Moreover,
let us introduce the first down-crossing time
\[\tau_b^-:=\inf\{t\ge0\,:\,X_t\le b\},\,\,\,\forall b\in\mathbb{R}.\]
If $x_h^\star<x_{h-1}^\star$, then the stopped process $(V_{t\wedge\tau_{x_h^\star}^-\wedge\tau_{x_{h-1}^\star}^{+}}^{(h)})_{t\ge0}$  is equal to a stopped supermartingale less a stopped martingale, and hence a supermartingale. Finally, for all $x>x_{h-1}^\star\ge x_h^\star$,
we have that
\[v^{(h)}(x)-v^{(h-1)}(x)=\phi^{(h)}(x)-\phi^{(h-1)}(x)=\ex_x[\mathrm{e}^{-\alpha\delta}[v^{(h-1)}(X_\delta)-v^{(h-2)}(X_\delta)]].\]
Hence, for all $t<\infty$ and $x>x_{h-1}^\star$, we have  
\bq
\ex_x[V_{t\wedge\tau_{x_{h-1}^\star}^-}^{(h)}]&=&\ex_x[\mathrm{e}^{-\alpha(\delta+t\wedge\tau_{x_{h-1}^\star}^-)}[v^{(h-1)}(X_{\delta+t\wedge\tau_{x_{h-1}^\star}^-})-v^{(h-2)}(X_{\delta+t\wedge\tau_{x_{h-1}^\star}^-})]]\nn\\
&\le&\ex_x[\mathrm{e}^{-\alpha\delta}[v^{(h-1)}(X_\delta)-v^{(h-2)}(X_\delta)]]=v^{(h)}(x)-v^{(h-1)}(x)=V_0^{(h)},\label{ineq}
\eq
where we used the assumption that $(V_t^{(h-1)})_{t\ge0}$ is a supermartingale, and the  independence between $\delta$ and $X$. Combining all cases, we conclude that the process $(V_t^{(h)})_{t\ge0}$ is a supermartingale.

The class (D) property of $(V_t^{(h)})_{t\ge0}$ now follows from  Minkowski's inequality and the fact that the elements in $\S^{(k)}$, $k=h-1,h$, are uniformly bounded in $\mathrm{L}^\varrho(\diff\pr_x)$ (see Proposition \ref{prop:regularity} above).

To finish the proof, we need to show that $\log K<x_{h+1}^\star\le x_{h}^\star$. To this end, we notice that  for all $x\ge x_h^\star$,
\begin{align}
v^{(h+1)}(x)=&\sup_{\tau\in\setT}\ex_{x}[\mathrm{e}^{-\alpha \tau}\phi^{(h+1)}(X_\tau)]\le\sup_{\tau\in\setT}\ex_{x}[\mathrm{e}^{-\alpha\tau}\phi^{(h)}(X_\tau)]+\sup_{\tau\in\setT}\ex_{x}[\mathrm{e}^{-\alpha\tau}[\phi^{(h+1)}(X_\tau)-\phi^{(h)}(X_\tau)]]\nn\\
=&v^{(h)}(x)+\sup_{\tau\in\setT}\ex_{x}[\mathrm{e}^{-\alpha(\tau+\delta)}[v^{(h)}(X_{\tau+\delta})-v^{(h-1)}(X_{\tau+\delta})]]=v^{(h)}(x)+\sup_{\tau\in\setT}\ex_{x}[V_{\tau+\delta}^{(h)}]\nn\\
\le&v^{(h)}(x)+\ex_{x}[V_{\delta}^{(h)}]=v^{(h)}(x)+[\phi^{(h+1)}(x)-\phi(x)]-[\phi^{(h)}(x)-\phi(x)]\nn\\
=&v^{(h)}(x)-\phi^{(h)}(x)+\phi^{(h+1)}(x)=\phi^{(h+1)}(x),\label{ineq2}
\end{align}
where we used the class ($D$) property of $(V_t^{(h)})_{t\ge0}$ in the second inequality and \eqref{single2} in the fourth equality.  This shows that  $v^{(h+1)}(x) = \phi^{(h+1)}(x)$. Since $x_{h+1}^\star=\sup\{x\le x_1^\star\,:\,\phi^{(h+1)}(y)=v^{(h+1)}(y),\,\,\forall y>x\}$, we can conclude from \eqref{ineq2} that $x_{h+1}^\star\le x_h^\star$.

 \subsection{Proof of Lemma \ref{lem:boundedness}}\label{Appen_1}
 For any $p,q>1$ satisfying $p^{-1}+q^{-1}=1$, and any sufficiently large $z>0$, we have  
 \begin{multline}
 \pr_x\left(\sup_{t\ge0}\,[\mathrm{e}^{-\alpha t}\phi({X_t})]>z\right)
 =\pr_x(\exists t\ge0, \mathrm{e}^{-\alpha t}\phi({X_t})>z)=\pr_x(\exists t\ge0, X_t-\alpha t>\log(z+K\mathrm{e}^{-\alpha t}))\nn\\
 \le\pr_x\bigg(\exists t\ge0, X_t-\alpha t>\frac{1}{p}\log(p z)+\frac{1}{q}\log(q K\mathrm{e}^{-\alpha t})\bigg)=\pr_x\bigg(\exists t\ge0, X_t-\alpha t+\frac{\alpha }{q}t>\frac{1}{p}\log (p z)+\frac{1}{q}\log (q K)\bigg)\nn\\
 =\pr_x\bigg(\sup_{0\le t<\infty}(X_t-\frac{\alpha}{p}t)>\frac{1}{p}\log (pz)+\frac{1}{q}\log (q K)\bigg)\sim\exp\bigg(-\tilde\rho_0 \Big[\frac{1}{p}\log (p z)+\frac{1}{q}\log (q K)-x \Big]\bigg)=\frac{\mathrm{e}^{\tilde{\rho}_0x}(qK)^{-\frac{\tilde\rho_0}{q}}}{(p z)^{\frac{\tilde\rho_0}{p}}}, \notag
 \end{multline}
 where we used Proposition 1.8 on page 259 of \cite{AsmussenBook} and $\tilde\rho_0$ is the smallest positive root of
 \[\psi({\tilde\rho_0})-\frac{\alpha }{p}{\tilde\rho_0}=0.\]
 
 It is now sufficient to show that it is possible to choose  $p>1$ such that $\tilde\rho_0>p$, and hence the random variable $(\sup_{t\ge0}[\mathrm{e}^{-\alpha t}(\mathrm{e}^{X_t}-K)^+])^\varrho$ has a finite expectation for $\varrho=\frac{1}{2}(1+\frac{\tilde{\rho_0}}{p})>1$. To this end, we show that $\psi(p)-\alpha<0$ for a sufficiently small $p>1$.  Indeed,  for all $0<\beta<p$, we have $0>\frac{\beta}{p}(\psi(p)-\alpha)\ge\psi(\beta)-\frac{\beta}{p}\alpha$, where we used Jensen's inequality $(\ex[Y])^{\frac{\beta}{p}}\ge\ex[Y^{\frac{\beta}{p}}]$ for positive random variable $Y=\mathrm{e}^{pX_1}$ in the last step. Hence, the smallest positive solution $\tilde{\rho}_0>p>1$.

 Let us first assume that $\psi(1) -\alpha<0$. Then for sufficiently small $p>1$, by the continuity of $\psi$ at 1, we have
 \[\psi(p)- \alpha<0.\]
 On the other hand, if  $\psi(1)-\alpha=0$, and $\psi'(1)<0$. Then for sufficiently small $p>1$, we have $\psi'(p)<0$, and
 \be\label{ineq1}
 \psi(p)-\alpha =\psi(p)-\psi(1)<\psi'(p)(p-1)<0.\nn
 \ee

\subsection{Proof of Proposition \ref{lapfirstpass} }
Let us define a new measure $\pr^{\Phi(\alpha)}$ by \eqref{change_of_measure} for $\kappa = \Phi(\alpha)$.
 Under this measure, $X$ is a L\'evy process with Laplace exponent \citep[Corollary 3.10]{Kyprianou2006}:
\be\psi_{\Phi(\alpha)}(\beta)=\psi(\beta+\Phi(\alpha)) -\alpha, \quad {\beta > - \Phi(\alpha)}.\label{change measure}\ee
 Then, for any  $t>0$, the change of measure yields  the expectation 
\bq
\ex_x\left[\mathrm{e}^{-\alpha \tau_b^+}\ind_{\{\tau_b^+<t\}}\right]
&=&\mathrm{e}^{\Phi(\alpha)x}\,\ex_x\left[\mathrm{e}^{-\alpha \tau_b^+ +\Phi(\alpha)(X_{\tau_b^+}-x)}\cdot\mathrm{e}^{-\Phi(\alpha)X_{\tau_b^+}}\,\ind_{\{\tau_b^+<t\}}\right]\nn\\
&=& \ex_x^{\Phi(\alpha)}\left[\mathrm{e}^{-\Phi(\alpha)(X_{\tau_b^+}-x)}\ind_{\{\tau_b^+<t\}}\right]. \label{value_at_a}
\eq
We now let $t\to\infty$ in  \eqref{value_at_a}. By applying the monotone convergence theorem to the left hand side of \eqref{value_at_a}, we obtain $\ex_x[\mathrm{e}^{-\alpha\tau_b^+}\ind_{\{\tau_b^+<\infty\}}]$.
Similarly, notice that the non-negative random variable in the expectation \eqref{value_at_a}  is bounded by 1, $\pr^{\Phi(\alpha)}_x$-a.s. we can
apply the bounded convergence theorem to obtain that
\[\ex_x\left[\mte^{-\alpha\tau_b^+}\ind_{\{\tau_b^+<\infty\}}\right]=\exph\left[\mte^{-\Phi(\alpha)X_{\tau_{b-x}^+}}\ind_{\{\tau_b^+<\infty\}}\right]<\infty.\]
Now  because for any $\beta\ge0$ we have $\mte^{-\alpha\tau_b^+-\beta(X_{\tau_b^+}-b)}\ind_{\{\tau_b^+<\infty\}}\le \mte^{-\alpha\tau_b^+}\ind_{\{\tau_{b}^+<\infty\}}$, $\pr_x$-a.s., the dominated convergence theorem yields that
\be
\ex_x\left[\mte^{-\alpha\tau_b^+-\beta(X_{\tau_b^+}-b)}\ind_{\{\tau_b^+<\infty\}}\right]=\mte^{\beta(b-x)}\exph\left[\mte^{-(\Phi(\alpha)+\beta)X_{\tau_{b-x}^+}}\ind_{\{\tau_{b-x}^+<\infty\}}\right].\label{value_at_a1}
\ee

The right hand side of \eqref{value_at_a1} can be computed using Lemma 1 of \cite{alili-kyp}. Precisely,  we have
\bq\label{ht}
\ex^{\Phi(\alpha)}[\mathrm{e}^{-(\Phi(\alpha)+\beta)X_{\tau_{b-x}^+}}\ind_{\{\tau_{b-x}^+<\infty\}}]&=&\lim_{q\downarrow 0}\frac{\ex^{\Phi(\alpha)}[\mathrm{e}^{-(\Phi(\alpha)+\beta)\overline{X}_{\mathbf{e}_q}}\ind_{\{\overline{X}_{\mathbf{e}_q}>b-x\}}]}{\ex^{\Phi(\alpha)}[\mathrm{e}^{-(\Phi(\alpha)+\beta)\overline{X}_{\mathbf{e}_q}}]}.
\eq
The law of $\overline{{X}}_{\mathbf{e}_q}$  under $\p^{\Phi(\alpha)}$
 can be extracted from \eqref{Xmax_mgf} and \eqref{Xmax_dist}.  More precisely, for any $q>0$, let $\tilde{\mathcal{I}}_q:=\{\tilde\rho_{i,q}\}_{1\le i\le |\tilde{\mathcal{I}}_q|}$ and $\tilde{\mathcal{J}}:=\{\tilde\eta_{j}\}_{1\le j\le |\tilde{\mathcal{J}}|}$ be, respectively, the roots to
\be\label{rhotil}\psi_{\Phi(\alpha)}(\tilde\rho_{i,q})=q\,\textrm{ and }\,\psi_{\Phi(\alpha)}(\tilde\eta_j)=\infty,\,\text{ s.t. }\,\Re\tilde\rho_{i,q},\, \Re\tilde\eta_j>0,\,\,\,\forall 1\le i\le |\tilde{\mathcal{I}}_q|,\, 1\le j\le|\tilde{\mathcal{J}}|, \ee
which are indexed in the same way as the elements of $\mathcal{I}_\alpha$ and $\mathcal{J}$.
Then, we infer from  \eqref{change measure} that   $\tilde{\eta}_1=\beta_0-\Phi(\alpha)$, $\tilde{\rho}_{1,q}<\tilde{\eta}_1$ and  $\Re\tilde{\rho}_{i,q}\ge\tilde{\eta}_1$ for all $i\ge2$. Similarly, we let $\tilde{\mathcal{I}}_0:=\{\tilde{\rho}_{i,0}\}_{1\le i\le |\tilde{\mathcal{I}}_0|}$ be the roots to
\be
\psi_{\Phi(\alpha)}(\tilde{\rho}_{i,0})=0, \,\,\,\text{s.t.}\,\,\,\Re\tilde{\rho}_{i,0}\ge0.\nn
\ee
From \eqref{I2} and \eqref{change measure} we deduce  that $\tilde{\mathcal{I}}_0+\Phi(\alpha)=\mathcal{I}_\alpha$ and $\tilde{\mathcal{J}}+\Phi(\alpha)=\mathcal{J}$, which means that
\be{\rho}_{i,\alpha}=\tilde\rho_{i,0}+\Phi(\alpha),\,\,\forall 1\le i\le |\mathcal{I}_\alpha|;\,\quad {\eta}_j=\tilde\eta_{j}+\Phi(\alpha),\,\,\,\forall 1\le j\le |\mathcal{J}|.\label{eq:rho2rhot}
\ee
By our assumption, $\tilde{\rho}_{i,0}$'s are distinct and
$\tilde{\rho}_{1,0}=\rho_{1, \alpha} - \Phi(\alpha) = 0$ and $\Re\tilde{\rho}_{i,0}\ge\Re\tilde{\eta}_1>0$ for all $2\le i\le |\mathcal{I}_\alpha|$.
Moreover, the fact that the roots $\rho_{i,\alpha}$'s are single implies that $\psi_{\Phi(\alpha)}'(\tilde{\rho}_{i,0})=\psi'(\rho_{i,\alpha})\neq 0$, and hence each branch of the mapping $q \mapsto \psi_{\Phi(\alpha)}^{-1}(q)=\tilde{\rho}_{i,q}\in\tilde{\mathcal{I}}_q$ is locally a diffeomorphism around 0. It follows that $\tilde{\rho}_{i,q}$'s are also distinct for all sufficiently small $q>0$.  It follows from \eqref{Xmax_mgf} and \eqref{Xmax_dist} that
\bq
\ex ^{\Phi(\alpha)}[\mathrm{e}^{-(\Phi(\alpha)+\beta)\overline{{X}}_{\mathbf{e}_q}}]&=& \prod_{i=1}^{|\tilde{\mathcal{I}}_q|} \frac{\tilde\rho_{i,q}}{\tilde\rho_{i,q}+\Phi(\alpha)+\beta}\prod_{j=1}^{|\tilde{\mathcal{J}}|} \bigg(1+\frac{\Phi(\alpha)+\beta}{\tilde\eta_j}\bigg),\,\,\,\forall \beta\ge0,\nn\\
\pr^{\Phi(\alpha)}(\overline{X}_{\mathbf{e}_q}\in \diff y)&=&\sum_{i=1}^{|\tilde{\mathcal{I}}_q|}\bigg(\prod_{\substack{j=1\\j\neq i}}^{|\tilde{\mathcal{I}}_q|}\frac{\tilde\rho_{j,q}}{\tilde\rho_{j,q}-\tilde\rho_{i,q}}\prod_{j=1}^{|\tilde{\mathcal{J}}|}\frac{\tilde\eta_j-\tilde\rho_{i,q}}{\tilde\eta_j}\bigg)\tilde\rho_{i,q}\mathrm{e}^{-\tilde\rho_{i,q}y}\,\diff y,\,\,\,\forall y> 0.\nn
\eq
As $q\downarrow0$, we have that $\tilde{\rho}_{i,q}\to\tilde{\rho}_{i,0}$, and in particular, $\tilde{\rho}_{1,q}\to \tilde{\rho}_{1,0}=0$ and $\lim_{q\downarrow0}\tilde{\rho}_{i,q}\neq 0$ for all $i\ge 2$.  As a result,  for $\beta>0$,
\bq
\lim_{q\downarrow 0}\frac{\tilde\rho_{1,q}}{\ex^{\Phi(\alpha)} [\mathrm{e}^{-(\Phi(\alpha)+\beta)\overline{X}_{\mathbf{e}_q}}] }&=&(\Phi(\alpha)+\beta)\cdot\prod_{i=2}^{|\tilde{\mathcal{I}}_0|}\frac{\tilde\rho_{i,0}+\Phi(\alpha)+\beta}{\tilde\rho_{i,0}}\prod_{j=1}^{|\tilde{\mathcal{J}}|}\frac{\tilde\eta_j}{\eta_j+\Phi(\alpha)+\beta}\nn\\
&=&(\rho_{1,\alpha}+\beta)\cdot\prod_{i=2}^{|\tilde{\mathcal{I}}_0|}\frac{\rho_{i,\alpha}+\beta}{\tilde\rho_{i,0}}\prod_{j=1}^{|\tilde{\mathcal{J}}|}\frac{\tilde\eta_j}{\eta_j+\beta},\label{cal01}
\eq
where we used
\eqref{eq:rho2rhot} in the second equality.
Moreover, for $x<b$ and $\beta\ge 0$, the ratio
\begin{align}
\frac{\ex^{\Phi(\alpha)}[\mathrm{e}^{-(\Phi(\alpha)+\beta)\overline{X}_{\mathbf{e}_q}}\ind_{\{\overline{X}_{\mathbf{e}_q}>b-x\}}]}{\tilde\rho_{1,q}}
=&\frac{1}{\tilde\rho_{1,q}}\int_{(b-x,\infty)}\mathrm{e}^{-(\Phi(\alpha)+\beta) y}\sum_{i=1}^{|\tilde{\mathcal{I}}_q|}\bigg(\prod_{\substack{j=1\\j\neq i}}^{|\tilde{\mathcal{I}}_q|}\frac{\tilde\rho_{j,q}}{\tilde\rho_{j,q}-\tilde\rho_{i,q}}\prod_{j=1}^{|\tilde{\mathcal{J}}|}\frac{\tilde\eta_j-\tilde\rho_{i,q}}{\tilde\eta_j}\bigg)
\tilde\rho_{i,q}\mathrm{e}^{-\tilde\rho_{i,q}y}\diff y\nn\\
=&\frac{1}{\tilde\rho_{1,q}}\sum_{i=1}^{|\tilde{\mathcal{I}}_q|}\bigg(\prod_{\substack{j=1\\j\neq i}}^{|\tilde{\mathcal{I}}_q|}\frac{\tilde\rho_{j,q}}{\tilde\rho_{j,q}-\tilde\rho_{i,q}}\prod_{j=1}^{|\tilde{\mathcal{J}}|}\frac{\tilde\eta_j-\tilde\rho_{i,q}}{\tilde\eta_j}\bigg)\cdot\frac{\tilde\rho_{i,q}}{\tilde\rho_{i,q}+\beta}\mathrm{e}^{-[\Phi(\alpha)+\beta+\tilde\rho_{i,q}](b-x)}\nn\\
=&\sum_{i=2}^{|\tilde{\mathcal{I}}_q|}\frac{1}{\tilde\rho_{1,q}-\tilde\rho_{i,q}}\bigg(\prod_{\substack{j=2\\j\neq i}}^{|\tilde{\mathcal{I}}_q|}\frac{\tilde\rho_{j,q}}{\tilde\rho_{j,q}-\tilde\rho_{i,q}}\prod_{j=1}^{|\tilde{\mathcal{J}}|}\frac{\tilde\eta_j-\tilde\rho_{i,q}}{\tilde\eta_j}\bigg)\cdot\frac{\tilde\rho_{i,q}}{\tilde\rho_{i,q}+\beta}\mathrm{e}^{-[\Phi(\alpha)+\beta+\tilde\rho_{i,q}](b-x)}\nn\\
&+\frac{1}{\tilde\rho_{1,q}+\Phi(\alpha)+\beta}\bigg(\prod_{j=2}^{|\tilde{\mathcal{I}}_q|}\frac{\tilde\rho_{j,q}}{\tilde\rho_{j,q}-\tilde\rho_{1,q}}\prod_{j=1}^{|\tilde{\mathcal{J}}|}\frac{\tilde\eta_j-\tilde\rho_{1,q}}{\tilde\eta_j}\bigg)\mathrm{e}^{-[\Phi(\alpha)+\beta+\tilde\rho_{1,q}](b-x)},\nn\end{align}
which, as $q\downarrow 0$, tends to
\begin{align}&\sum_{i=2}^{|\tilde{\mathcal{I}}_0|}\bigg(\prod_{\substack{j=2\\j\neq i}}^{|\tilde{\mathcal{I}}_0|}\frac{\tilde\rho_{j,0}}{\tilde\rho_{j,0}-\tilde\rho_{i,0}}\prod_{j=1}^{|\tilde{\mathcal{J}}|}\frac{\tilde\eta_j-\tilde\rho_{i,0}}{\tilde\eta_j}\bigg)\cdot\frac{-1}{\tilde\rho_{i,0}+\Phi(\alpha)+\beta}\mathrm{e}^{-[\Phi(\alpha)+\beta+\tilde\rho_{i,0}](b-x)}+\frac{1}{\Phi(\alpha)+\beta}\mathrm{e}^{-(\Phi(\alpha)+\beta)(b-x)}\nn\\
=&\sum_{i=2}^{|\tilde{\mathcal{I}}_0|}\bigg(\prod_{\substack{j=2\\j\neq i}}^{|\tilde{\mathcal{I}}_0|}\frac{\tilde\rho_{j,0}}{\rho_{j,\alpha}-\rho_{i,\alpha}}\prod_{j=1}^{|\tilde{\mathcal{J}}|}\frac{\eta_j-\rho_{i,\alpha}}{\tilde\eta_j}\bigg)\cdot\frac{-1}{\rho_{i,\alpha}+\beta}\mathrm{e}^{-(\rho_{i,\alpha}+\beta)(b-x)}+\frac{1}{\rho_{1,\alpha}+\beta}\mathrm{e}^{-(\rho_{1,\alpha}+\beta)(b-x)}.\label{cal02}
\end{align}
Combining \eqref{value_at_a1}, \eqref{ht}, \eqref{cal01} and \eqref{cal02}, we obtain that, for all $\beta\ge0$,
\begin{align}
&\ex_{x}[\mathrm{e}^{-\alpha\tau_b^+-\beta (X_{\tau_b^+}-b)}\ind_{\{\tau_b^+<\infty\}}]\nn\\
= \,& \sum_{i=2}^{|\mathcal{I}_{\alpha}|}\bigg(\prod_{\substack{j=2\\j\neq i}}^{|\mathcal{I}_\alpha|}\frac{\rho_{j,\alpha}+\beta}{\rho_{j,\alpha}-\rho_{i,\alpha}}\prod_{j=1}^{|\mathcal{J}|}\frac{\eta_j-\rho_{i,\alpha}}{\eta_j+\beta}\bigg)\cdot\frac{-\Phi(\alpha)-\beta}{\tilde\rho_{i,0}}\mathrm{e}^{-\rho_{i,\alpha}(b-x)}+\prod_{j=2}^{|\mathcal{I}_{\alpha}|}\frac{\rho_{j,\alpha}+\beta}{\rho_{j,\alpha}-\rho_{1,\alpha}}\prod_{j=1}^{|\mathcal{J}|}\frac{\eta_j-\rho_{1,\alpha}}{\eta_j+\beta}\,\mte^{-\rho_{1,\alpha}(b-x)}\nn\\
=\,&\sum_{i=2}^{|\mathcal{I}_{\alpha}|}\bigg(\prod_{\substack{j=1\\j\neq i}}^{|\mathcal{I}_{\alpha}|}\frac{\rho_{j,\alpha}+\beta}{\rho_{j,\alpha}-\rho_{i,\alpha}}\prod_{j=1}^{|\tilde{\mathcal{J}}|}\frac{\tilde\eta_j-\tilde\rho_{i,0}}{\tilde\eta_j+\beta}\bigg)\mathrm{e}^{-\rho_{i,\alpha}(b-x)}+\prod_{j=2}^{|\mathcal{I}_{\alpha}|}\frac{\rho_{j,\alpha}+\beta}{\rho_{j,\alpha}-\rho_{1,\alpha}}\prod_{j=1}^{|\mathcal{J}|}\frac{\eta_j-\rho_{1,\alpha}}{\eta_j+\beta}\,\mte^{-\rho_{1,\alpha}(b-x)}\nn\\
=\,&\sum_{i=1}^{|\mathcal{I}_{\alpha}|}\bigg(\prod_{\substack{j=1\\j\neq i}}^{|\mathcal{I}_{\alpha}|}\frac{\rho_{j,\alpha}+\beta}{\rho_{j,\alpha}-\rho_{i,\alpha}}\prod_{j=1}^{|\mathcal{J}|}\frac{\eta_j-\rho_{i,\alpha}}{\eta_j+\beta}\bigg)\mathrm{e}^{-\rho_{i,\alpha}(b-x)}=\frac{1}{\psi_\alpha^+(\beta)}\sum_{i=1}^{|\mathcal{I}_\alpha|}A_i\frac{\rho_{i,\alpha}}{\rho_{i,\alpha}+\beta}\mte^{-\rho_{i,\alpha}(b-x)}.\label{mgf_x_bar}
\end{align}
This completes the proof.

\subsection{Proof of Theorem \ref{thm:value_single_stop}}
We only need to prove the assertion for the case $\alpha\le0$ since the case of non-negative discount rate $\alpha> 0$ has been addressed in \cite{mordecki2002}.

We begin by differentiating  $g^{(1)}(x,b)$ with respect to $b>x\vee\log K$ to get 
\be
\frac{\partial}{\partial b}g^{(1)}(x,b)=\frac{-\mathrm{e}^{b}+K\psi_\alpha^+(-1)}{\psi_\alpha^+(-1)}\cdot\sum_{i=1}^{|\mathcal{I}_\alpha|}A_i\rho_{i,\alpha}\mathrm{e}^{-\rho_{i,\alpha}(b-x)}, \quad \forall x < b.\label{cal5}\ee
Clearly, $x_1^\star=\log(K\psi_\alpha^+(-1))$ satisfies the first order condition $ {\partial}g^{(1)}(b,x) / {\partial b} = 0$. To show that $x_1^\star$ is indeed the optimal exercise threshold, we only need to verify  the followings (see, for example, \cite{xiazhou}):
\begin{enumerate}
\item  for $x\ge x_1^\star$, $g^{(1)}(x,\cdot)$ is  decreasing for all $b> x$ and $\sup_{b\le x}g^{(1)}(x,b)=\phi(x)\ge\lim_{b\downarrow x}g^{(1)}(x,b)$; 
\item for $x<x_1^\star$, $g^{(1)}(x,\cdot)$ is increasing for all $x<b\le x_1^\star$ and is non-increasing for all $b\ge x_1^\star$, and $\sup_{b\le x}g^{(1)}(x,b)=\phi(x)\le \lim_{b\downarrow x}g^{(1)}(x,b)$.
\end{enumerate}
Since $\psi_\alpha^+(-1)>0$, it follows from    \eqref{cal5}  that the   monotonicity of $g^{(1)}(x,\cdot)$   {for $b>x$} amounts to showing that 
\[\sum_{i=1}^{|\mathcal{I}_\alpha|}A_i\rho_{i,\alpha}\mathrm{e}^{-\rho_{i,\alpha}y}\ge0,\,\,\,\forall y>0.\]
By setting $\beta=0$ in  \eqref{cal01},  we obtain
\be
\lim_{q\downarrow0}\frac{\tilde{\rho}_{1,q}}{\ex^{\Phi(\alpha)}[\mathrm{e}^{-\Phi(\alpha)\overline{X}_{\mathbf{e}_q}}]}=\rho_{1,\alpha}\cdot\prod_{i=2}^{|\mathcal{I}_\alpha|}\frac{\rho_{i,\alpha}}{\rho_{i,\alpha}-\rho_{1,\alpha}}\prod_{j=1}^{|\mathcal{J}|}\frac{\eta_j-\rho_{1,\alpha}}{\eta_j}=\rho_{1,\alpha}\cdot A_1.\label{cal3}
\ee
Similarly, it follows from  \eqref{cal02} that,  for all $y>0$,
\bq
\frac{1}{\diff y}\lim_{q\downarrow0}\frac{\mathrm{e}^{-\Phi(\alpha)y}\pr^{\Phi(\alpha)}(\overline{X}_{\mathbf{e}_q}\in\diff y)}{\tilde{\rho}_{1,q}}&=&-\sum_{i=2}^{|\mathcal{I}_\alpha|}\bigg(\prod_{\substack{j=2\\j\neq i}}^{|\mathcal{I}_\alpha|}\frac{\rho_{j,\alpha}-\rho_{1,\alpha}}{\rho_{j,\alpha}-\rho_{i,\alpha}}\prod_{j=1}^{|\mathcal{J}|}\frac{\eta_j-\rho_{i,\alpha}}{\eta_j-\rho_{1,\alpha}}\bigg)\mathrm{e}^{-\rho_{i,\alpha}y}+\mathrm{e}^{-\rho_{1,\alpha}y}\nn\\
&=&\sum_{i=2}^{|\mathcal{I}_\alpha|}\frac{\rho_{i,\alpha}\cdot A_{i}}{\rho_{1,\alpha}\cdot A_1}\mathrm{e}^{-\rho_{i,\alpha}y}+\mathrm{e}^{-\rho_{1,\alpha}y}
= \sum_{i=1}^{|\mathcal{I}_\alpha|}\frac{\rho_{i,\alpha}\cdot A_{i}}{\rho_{1,\alpha}\cdot A_1}\mathrm{e}^{-\rho_{i,\alpha}y}.\label{cal4}
\eq
From \eqref{cal3} and \eqref{cal4}, we obtain
\be
\sum_{i=1}^{|\mathcal{I}_\alpha|}A_i\rho_{i,\alpha}\mathrm{e}^{-\rho_{i,\alpha}y}=\frac{\mathrm{e}^{-\Phi(\alpha)y}}{\diff y}\lim_{q\downarrow0}\frac{\pr^{\Phi(\alpha)}(\overline{X}_{\mathbf{e}_q}\in\diff y)}{\ex^{\Phi(\alpha)}[\mathrm{e}^{-\Phi(\alpha)\overline{X}_{\mathbf{e}_q}}]}\ge 0,\,\,\,\forall y>0.\label{eq:Arho}
\ee
To complete the proof that $x_1^\star$ is indeed the optimal exercise threshold, we need to show that, for any $x\ge x_1^\star$, $\phi(x)\ge \lim_{b\downarrow x}g^{(1)}(x,b)$; and for $x<x_1^\star$, $\phi(x)\le \lim_{b\downarrow x}g^{(1)}(x,b)$. To this end, notice that $\phi(x)=g^{(1)}(x,x_1^\star)$ for all $x\ge x_1^\star$. On the other hand,
 using Corollary \ref{value_at_passage} we have that 
\bq\lim_{b\downarrow x}g^{(1)}(x,b)=\bigg(1-\frac{\psi_\alpha^+(\infty)}{\psival}\bigg)\mte^x-K(1-\psi_\alpha^+(\infty))=\phi(x)+\psi_\alpha^+(\infty)\bigg(\frac{\mte^{x_1^\star}-\mte^x}{\psival}\bigg)=\left\{\begin{array}{ll}\le\phi(x), ~&\text{if }x\ge x_1^\star\\
\ge\phi(x), ~&\text{if }x< x_1^\star\end{array}\right..\nn\eq
Thus, $b_k$ is indeed the optimal exercise threshold for any $x\in\R$.  Finally, \eqref{valone} follows from \eqref{first passage} by setting $b=x_1^\star$.

 \subsection{Proof of Proposition \ref{thm:first order}}
 First, notice that   \eqref{first passage} and \eqref{delayed payoff} imply that, for any fixed $x\in\R$, the function $g^{(k)}(x,b)$   is differentiable in $b$ for all $b> x\vee\log K$. 
 Direct calculation (using \eqref{nudensity}) gives the derivative
 \begin{align}
&\frac{\partial}{\partial b}g^{(k)}(x,b)\nn\\
=&\sum_{i=1}^{|\mathcal{I}_\alpha|}A_i\rho_{i,\alpha}\mathrm{e}^{-\rho_{i,\alpha}(b-x)}\bigg(\frac{\mathrm{e}^{x_1^\star}-\mathrm{e}^{b}}{\psi_\alpha^+(-1)}+\frac{1}{\phi_\infty}
\ex [\mathrm{e}^{-\alpha\delta}({v^{(k-1)}_+}'(b+X_\delta)-\frac{\rho_{i,\alpha}}{\phi_\infty}v^{(k-1)}(b+X_\delta))]\bigg)\nn\\
&-\bigg[\sum_{i=1}^{|\mathcal{I}_\alpha|}A_i\rho_{i,\alpha}\mathrm{e}^{-\rho_{i,\alpha}(b-x)}\bigg(\ex [\mathrm{e}^{-\alpha\delta}v^{(k-1)}(b+X_\delta)](-\frac{\rho_{i,\alpha}}{\phi_\infty}+\nu(\{0\}))-\int_{(0,\infty)} \ex [\mathrm{e}^{-\alpha\delta}v^{(k-1)}(b+y+X_\delta)]\nu(\diff y)\bigg)\bigg]\nn\\
=&\bigg[\frac{\mathrm{e}^{x_1^\star}-\mathrm{e}^{b}}{\psi_\alpha^+(-1)}+\bigg(\frac{\ex[\mathrm{e}^{-\alpha\delta}{v^{(k-1)}_+}'(b+X_\delta)]}{\phi_\infty}-\int_{[0,\infty)} \ex [\mathrm{e}^{-\alpha\delta}v^{(k-1)}(b+y+X_\delta)]\nu(\diff y)\bigg)\bigg]\times\bigg(\sum_{i=1}^{|\mathcal{I}_\alpha|}A_i\rho_{i,\alpha}\mathrm{e}^{-\rho_{i,\alpha}(b-x)}\bigg)\nn\\
=&\bigg[\frac{\mathrm{e}^{x_1^\star}-\mathrm{e}^{b}}{\psi_\alpha^+(-1)}+\ex[\mte^{-\alpha\delta}u^{(k-1)}(b+X_\delta)]\bigg]\times\bigg(\sum_{i=1}^{|\mathcal{I}_\alpha|}A_i\rho_{i,\alpha}\mathrm{e}^{-\rho_{i,\alpha}(b-x)}\bigg).\label{eq:deri}
\end{align}

Recall the inequality    \eqref{eq:Arho} in the case with  $\alpha\le0$. For $\alpha>0$, we compute  from \eqref{Xmax_dist} to get \[\sum_{i=1}^{|\mathcal{I}_\alpha|}A_i\rho_{i,\alpha}\mathrm{e}^{-\rho_{i,\alpha}y}=\frac{1}{\diff y}\pr(\overline{X}_{\ep}\in\diff y)\ge0,\,\,\,\forall y> 0.\] 
Notice that, due to the linear independence of $\mte^{-\rho_{i,\alpha}y}$'s,  the left hand side of the above equation is strictly positive on all but a possibly  finite set in $\R_+$. Moreover, for all $x\ge x_k^\star$, we have $g^{(k)}(x,x_k^\star) \ge \phi^{(k)}(x)=v^{(k)}(x)\ge g^{(k)}(x,b)$  for all $b\ge x$, hence $\frac{\partial}{\partial b}|_{b=x+}g^{(k)}(x,b)\le 0$. This  implies that
\[\frac{\mathrm{e}^{x_1^\star}-\mathrm{e}^{b}}{\psi_\alpha^+(-1)}+\ex[\mte^{-\alpha\delta}u^{(k-1)}(b+X_\delta)]\le 0,\,\,\,\forall b\ge x_k^\star.\]
On the other hand, for all $x\le\log K$, from the proof of Proposition \ref{x_k} we know that $g^{(k)}(x,x_{k-1}^\star)\ge g^{(k)}(x,x)=\phi^{(k)}(x)$. It follows that there exists at least a $b\in[x,x_{k-1}^\star]\subset[x, x_1^\star]$ such that $\frac{\partial}{\partial b}g^{(k)}(x,b)> 0$, and hence
\[\frac{\mathrm{e}^{x_1^\star}-\mathrm{e}^{b}}{\psi_\alpha^+(-1)}+\ex[\mte^{-\alpha\delta}u^{(k-1)}(b+X_\delta)]> 0.\]
By the assumed continuity of $u^{(k-1)}$, we know that there exists at least one solution to \eqref{first order condition}. If the optimal exercise threshold $b_k^\star$ exists, then we know that $b_k^\star\le x_k^\star\le x_1^\star<\infty$. For any  fixed $x<b_k^\star$, the function $g^{(k)}(x,b)$ is maximized at $b=b_k^\star$, hence $\frac{\partial}{\partial b}|_{b=b_k^\star}g^{(k)}(x,b)=0$ for all $x<b_k^\star$. This implies that $b_k^\star$ is a solution to \eqref{first order condition}.

\subsection{Proof of Proposition \ref{prop:ukk}}
We will first prove an auxiliary lemma connecting measures $\{\bar{\nu}_i(\diff y)\}_{i=1}^{|\mathcal{I}_\alpha|}$ with $\nu(\diff y)$.
\begin{lemma}\label{lem:nu}
Let $(\bar{\nu}_i)_+'(z)$ and $\nu_+'(y)$ be the right derivatives of $\bar{\nu}_i[0,z)$ and $\nu[0,y)$, respectively. Then
\be
-\sum_{i=1}^{|\mathcal{I}_\alpha|}\frac{A_i}{\rho_{i,\alpha}}[(\bar{\nu}_i)_+'(z)][(\bar{\nu}_i)_+'(y)]=\nu_+'(z+y),\,\,\,\forall y,z>0.\label{nunu}
\ee
\end{lemma}
\begin{proof}
We will prove \eqref{nunu}  by using the bivariate Laplace transform. To this end, let $\beta_1,\beta_2\ge0$ and that $\beta_1\neq\beta_2$, then by \eqref{nudensity} and \eqref{nubar} we have
\begin{align}
&-\int_{(0,\infty)}\int_{(0,\infty)}\mte^{-\beta_1 y-\beta_2 z}\sum_{i=1}^{|\mathcal{I}_\alpha|}\frac{A_i}{\rho_{i,\alpha}}[(\bar{\nu}_i)_+'(y)][(\bar{\nu}_i)_+'(z)]\diff y\diff z\nn\\
=&-\sum_{i=1}^{|\mathcal{I}_\alpha|}A_i\rho_{i,\alpha}\bigg(\frac{1}{\rho_{i,\alpha}+\beta_1}\frac{1}{\psi_\alpha^+(\beta_1)}-\frac{1}{\phi_\infty}\bigg)\bigg(\frac{1}{\rho_{i,\alpha}+\beta_2}\frac{1}{\psi_\alpha^+(\beta_2)}-\frac{1}{\phi_\infty}\bigg)\nn\\
=&\frac{1/(\beta_2-\beta_1)}{\psi_\alpha^+(\beta_1)\psi_\alpha^+(\beta_2)}\sum_{i=1}^{|\mathcal{I}_\alpha|}A_i\bigg(\frac{\rho_{i,\alpha}}{\rho_{i,\alpha}+\beta_2}-\frac{\rho_{i,\alpha}}{\rho_{i,\alpha}+\beta_1}\bigg)+\sum_{i=1}^{|\mathcal{I}_\alpha|}\frac{A_i}{\phi_\infty}\bigg(\frac{1}{\psi_\alpha^+(\beta_1)}\frac{\rho_{i,\alpha}}{\rho_{i,\alpha}+\beta_1}+\frac{1}{\psi_\alpha^+(\beta_2)}\frac{\rho_{i,\alpha}}{\rho_{i,\alpha}+\beta_2}-\frac{\rho_{i,\alpha}}{\phi_\infty}\bigg)\nn\\
=&\frac{1}{\beta_2-\beta_1}\frac{\psi_\alpha^+(\beta_2)-\psi_\alpha^+(\beta_1)}{\psi_\alpha^+(\beta_1)\psi_\alpha^+(\beta_2)}+\frac{1}{\phi_\infty}\bigg(1-\frac{\psi_\alpha^+(\infty)}{\psi_\alpha^+(\beta_1)}+1-\frac{\psi_\alpha^+(\infty)}{\psi_\alpha^+(\beta_2)}\bigg)-\frac{1}{\phi_\infty^2}\sum_{i=1}^{|\mathcal{I}_\alpha|}A_i\rho_{i,\alpha}\nn\\
=&\frac{1}{\beta_2-\beta_1}\bigg(\frac{1}{\psi_\alpha^+(\beta_1)}-\frac{1}{\psi_\alpha^+(\beta_2)}\bigg)+\frac{1}{\phi_\infty}\bigg(2-\frac{\psi_\alpha^+(\infty)}{\psi_\alpha^+(\beta_1)}-\frac{\psi_\alpha^+(\infty)}{\psi_\alpha^+(\beta_2)}-\frac{1}{\phi_\infty}\sum_{i=1}^{|\mathcal{I}_\alpha|}A_i\rho_{i,\alpha}\bigg)\nn\\
=&\frac{1}{\beta_2-\beta_1}\bigg(\frac{1}{\psi_\alpha^+(\beta_1)}-\frac{1}{\psi_\alpha^+(\beta_2)}\bigg)+\frac{1}{\phi_\infty},\label{LAP1}\end{align}
where,
in the last equality, we used that fact that, if $-J$ is a subordinator, then $\phi_\infty=\infty$; otherwise, we have $\psi_\alpha^+(\infty)=0$, and 
\[\sum_{i=1}^{|\mathcal{I}_\alpha|}A_i\rho_{i,\alpha}=\lim_{\beta \rightarrow \infty}\bigg(\beta\cdot\sum_{i=1}^{|\mathcal{I}_\alpha|}A_i\frac{\rho_{i,\alpha}}{\rho_{i,\alpha}+\beta}\bigg)=\lim_{\beta\to\infty}\beta\psi_\alpha^+(\beta)=\phi_\infty.\]
On the other hand, by \eqref{eq:numeasure} we have, for $\beta_1,\beta_2\ge0$ such that $\beta_1\neq\beta_2$,
\begin{align}
&\int_{(0,\infty)}\int_{(0,\infty)}\mte^{-\beta_1y-\beta_2z}\nu_+'(y+z)\diff y\diff z
\stackrel{s=y+z}{=}\int_{(0,\infty)}\mte^{-\beta_2s}\nu_+'(s)\diff s\int_{(0,s)}\mte^{(\beta_2-\beta_1)y}\diff y\nn\\
=&\int_{(0,\infty)}\mte^{-\beta_2s}\frac{\mte^{(\beta_2-\beta_1)s}-1}{\beta_2-\beta_1}\nu_+'(s)\diff s=\frac{1}{\beta_2-\beta_1}\bigg(\int_{[0,\infty)}\mte^{-\beta_1s}\nu(\diff s)-\nu(\{0\})-\int_{[0,\infty)}\mte^{-\beta_2s}\nu(\diff s)+\nu(\{0\})\bigg)\nn\\
=&\frac{1}{\beta_2-\beta_1}\bigg(\frac{1}{\psi_\alpha^+(\beta_1)}-\frac{\beta_1}{\phi_\infty}-\frac{1}{\psi_\alpha^+(\beta_2)}+\frac{\beta_2}{\phi_\infty}\bigg)=\frac{1}{\beta_2-\beta_1}\bigg(\frac{1}{\psi_\alpha^+(\beta)}-\frac{1}{\psi_\alpha^+(\beta_2)}\bigg)+\frac{1}{\phi_\infty}.\label{LAP2}
\end{align}
From \eqref{LAP1} and \eqref{LAP2} we know that \eqref{nunu} holds for all $y,z>0$.
\end{proof}

We are now ready to give the proof of Proposition \ref{prop:ukk}.  Since threshold-type strategies are optimal for problem \eqref{single1} with  {up to} $(k-1)$ exercise opportunities by assumption, it follows   that $v^{(k-1)}(x)=g^{(k-1)}(x,x_{k-1}^\star)$ for all $x\in\R$, and  $u^{(k-1)}(x)=\tilde{u}^{(k-1)}(x)$ for all $x\in\R$ by comparing \eqref{function:u} and  \eqref{tildeuk}. Also,  observe that $g^{(k)}(x,b_k)=\phi^{(k)}(x)$ for  $x\ge b_k>\log K$. Applying this fact to  \eqref{tildeuk}, we  get
\begin{align}
\tilde{u}^{(k)}(x)=&\frac{{\phi^{(k)}}'(x)}{\phi_\infty}-\int_{[0,\infty)}\phi^{(k)}(x+y)\nu(\diff y)\nn\\
=&\frac{\mathrm{e}^x+\ex_x[\mathrm{e}^{-\alpha\delta}{v^{(k-1)}_+}'(X_\delta)]}{\phi_\infty}-\int_{[0,\infty)}(\mathrm{e}^{x+y}-K+\ex_x[\mathrm{e}^{-\alpha\delta}v^{(k-1)}(y+X_\delta)])\nu(\diff y)\nn\\
=&\mathrm{e}^x\bigg[\frac{1}{\phi_\infty}-\bigg(\frac{1}{\psi_\alpha^+(-1)}-\frac{(-1)}{\phi_\infty}\bigg)\bigg]+K+\ex\bigg[\mathrm{e}^{-\alpha\delta}\bigg(\frac{{v^{(k-1)}_+}'(x+X_\delta)}{\phi_\infty}-\int_{[0,\infty)} v^{(k-1)}(x+y+X_\delta)\nu(\diff y)\bigg)\bigg]\nn\\
=&\frac{\mathrm{e}^{x_1^\star}-\mathrm{e}^x}{\psi_\alpha^+(-1)}+\ex[\mathrm{e}^{-\alpha\delta}u^{(k-1)}(x+X_\delta)]=\tilde{u}_0^{(k)}(x).\label{tildeulast}
\end{align}

For $x<b_k$, we use \eqref{eq:touch} and \eqref{eq:overshoot} to write
\be
g^{(k)}(x,{b_k})=\ex_x[\mathrm{e}^{-\alpha\tau_{{b_k}}^+}\phi^{(k)}(X_{\tau_{{b_k}}^+})\ind_{\{\tau_{{b_k}}^+<\infty\}}]=\sum_{i=1}^{|\mathcal{I}_\alpha|}A_i\mathrm{e}^{\rho_{i,\alpha}(x-{b_k})}\int_{[0,\infty)}\phi^{(k)}(b_k+y)\bar{\nu}_i(\diff y).\nn
\ee
It follows that,
\begin{align}\frac{1}{\phi_\infty}\frac{\partial}{\partial x}g^{(k)}(x,b_k)=&\int_{[0,\infty)}\phi^{(k)}(b_k+y)\bigg(\sum_{i=1}^{|\mathcal{I}_\alpha|}A_i\mte^{-\rho_{i,\alpha}(b_k-x)}\frac{\rho_{i,\alpha}}{\phi_\infty}\bar{\nu}_i(\diff y)\bigg),\label{Es1}\\
\int_{[0,b_k-x)}g^{(k)}(x+z,b_k)\nu(\diff z)=&\int_{[0,b_k-x)}\int_{[0,\infty)}\phi^{(k)}(b_k+y)\bigg(\sum_{i=1}^{|\mathcal{I}_\alpha|}A_i\mte^{-\rho_{i,\alpha}(b_k-x-z)}\bar{\nu}_i(\diff y)\bigg)\nu(\diff z)\nn\\
=&\int_{[0,\infty)}\phi^{(k)}(b_k+y)\bigg[\sum_{i=1}^{|\mathcal{I}_\alpha|}A_i\mte^{-\rho_{i,\alpha}(b_k-x)}\bigg(\int_{[0,b_k-x)}\mte^{\rho_{i,\alpha}z}\nu(\diff z)\bigg)\bar{\nu}_i(\diff y)\bigg].\label{Es2}
\end{align}
Moreover, from Lemma \ref{lem:nu} we know that 
\begin{align}
\sum_{i=1}^{|\mathcal{I}_\alpha|}A_i\mte^{-\rho_{i,\alpha}(b_k-x)}\bigg(\frac{\rho_{i,\alpha}}{\phi_\infty}-\int_{[0,b_k-x)}\mte^{\rho_{i,\alpha}z}\nu(\diff z)\bigg)\bar{\nu}_i(\diff y)=-\sum_{i=1}^{|\mathcal{I}_\alpha|}\frac{A_i}{\rho_{i,\alpha}}[(\bar{\nu}_i)_+'(b_k-x)][(\bar{\nu}_i)_+'(y)]\diff y=\nu(b_k-x+\diff y).\label{Es3}
\end{align}
From \eqref{Es1}, \eqref{Es2} and \eqref{Es3} we have 
\begin{align}
\tilde{u}^{(k)}(x)=&\frac{1}{\phi_\infty}\frac{\partial}{\partial x}g^{(k)}(x,b_k)-\int_{[0,b_k-x)}g^{(k)}(x+z,b_k)\nu(\diff z)-\int_{[b_k-x,\infty)}g^{(k)}(x+y,b_k)\nu(\diff y)\nn\\
=&\int_{[0,\infty)}\phi^{(k)}(b_k+y)\nu(b_k-x+\diff y)-\int_{[b_k-x,\infty)}\phi^{(k)}(x+y)\nu(\diff y)\nn\\
=&\int_{[b_k-x,\infty)}\phi^{(k)}(x+y)\nu(\diff y)-\int_{[b_k-x,\infty)}\phi^{(k)}(x+y)\nu(\diff y)=0.
\end{align}
As a result, we have  $\tilde{u}^{(k)}(x)\equiv0$ for all $x<{b_k}$. Combining this with \eqref{tildeulast} yields \eqref{tildeu2}.  
Moreover, from the definition of $\tilde{u}_0^{(k)}(x)$ in \eqref{first order condition}, we know that  $\tilde{u}^{(k)}$ is  non-increasing on $[b_k,\infty)$. 
The first order condition \eqref{first order condition} shows that $\tilde{u}^{(k)}$ is also continuous at ${b_k}$, and thus $\tilde{u}^{(k)}$ is  non-positive on $\R$.

 \subsection{Proof of Lemma \ref{uniqueness} } 
Suppose there are two distinct solutions to \eqref{first order condition}, say, ${b_k}<{b_k}'$. Then by Proposition \ref{prop:ukk} we have a non-increasing, non-positive function
\[\ind_{\{x\ge {b_k}\}}\left(\frac{\mathrm{e}^{x_1^\star}-\mathrm{e}^{x}}{\psi_\alpha^+(-1)}+\ex[\mathrm{e}^{-\alpha\delta}u^{(k-1)}(x+X_\delta)]\right),\,\,\,\forall x\in\R.\]
Moreover, this function is strictly decreasing for all $x\ge{b_k}$. This implies that
\[\frac{\mathrm{e}^{x_1^\star}-\mathrm{e}^{{b_k}'}}{\psi_\alpha^+(-1)}+\ex[\mathrm{e}^{-\alpha\delta}u^{(k-1)}({b_k}'+X_\delta)]<0.\]
 This contradicts the assumption that  ${b_k}'$ solves \eqref{first order condition}. 

To finish the proof, we let ${b_k}$ be the unique solution to \eqref{first order condition}. Notice from \eqref{eq:deri} that
\begin{enumerate}
\item for all fixed $x\ge{b_k}$, the function $g^{(k)}(x,\cdot)$ is decreasing in $b$ for all $b>x$;
\item for all fixed $x<{b_k}$, the function $g^{(k)}(x,\cdot)$ is increasing in $b$ for all $x<b\le {b_k}$, and decreasing in $b$ for all $b\ge{b_k}$.
\end{enumerate}

We now apply  a similar argument as in the proof of Theorem \ref{thm:value_single_stop} to show that, for all $x\ge b_k$, $\lim_{b\downarrow x}g^{(k)}(x,b)\le \phi^{(k)}(x)$; and for all $x<b_k$, $\lim_{b\downarrow x}g^{(k)}(x,b)\ge \phi^{(k)}(x)$. This will allow us to conclude that $b_k$ is indeed the optimal exercise threshold. To this end, from \eqref{Xmax_mgf}, \eqref{eq:psi-alpha}, \eqref{eq:numeasure}, \eqref{nubar} and the fact that $\psi_\alpha^+(\infty)\cdot\phi_\infty=0$, we know that
\be\sum_{i=1}^{|\mathcal{I}_\alpha|}A_i\bar{\nu}_i(\diff y)=
\ind_{\{y=0\}}-\psi_\alpha^+(\infty)\nu(\diff y).\label{ANU}\ee
Below we  consider  two cases separately:
\begin{enumerate}
\item if $-J$ is not a subordinator, then $\psi_\alpha^+(\infty)=0$ and $\phi_\infty>0$. 
Using \eqref{eq:touch} and \eqref{eq:overshoot} and the monotone convergence theorem, we obtain  the limit  for any fixed $x\in\R$:
\be
\lim_{b\downarrow x}g^{(k)}(x,b)=\int_{[0,\infty)}\phi^{(k)}(x+y)\bigg(\sum_{i=1}^{|\mathcal{I}_\alpha|}A_i\bar{\nu}_i(\diff y)\bigg)=\phi^{(k)}(x),\label{eq:glimit}\ee
where we used \eqref{ANU} in the last equality.
\item if $-J$ is a subordinator, then $\psi_\alpha^+(\infty)>0$ and $\phi_\infty=0$.  
Similarly to \eqref{eq:glimit}, we use \eqref{ANU} to obtain that,  for any fixed $x\in\R$,
\begin{align}
\lim_{b\downarrow x}g^{(k)}(x,b)=&\int_{[0,\infty)}\phi^{(k)}(x+y)\bigg(\sum_{i=1}^{|\mathcal{I}_\alpha|}A_i\bar{\nu}_i(\diff y)\bigg)=\phi^{(k)}(x)-\psi_\alpha^+(\infty)\int_{[0,\infty)}\phi^{(k)}(x+y)\nu(\diff y)\nn\\
=&\phi^{(k)}(x)-\psi_\alpha^+(\infty)\int_{[0,\infty)}\left(\mte^{x+y}-K+\ex[\mte^{-\alpha\delta}v^{(k-1)}(x+y+X_\delta)]\right)\nu(\diff y)\nn\\
=&\phi^{(k)}(x)+\psi_\alpha^+(\infty)\bigg(\frac{\mte^{x_1^\star}-\mte^{x}}{\psi_\alpha^+(-1)}+\ex[\mte^{-\alpha\delta}u^{(k-1)}(x+X_\delta)]\bigg)\nn\\
=&\phi^{(k)}(x)+\psi_\alpha^+(\infty)\tilde{u}_0^{(k)}(x)=\left\{\begin{array}{ll}\le\phi^{(k)}(x),\,\,&\text{if }x\ge b_k\\ >\phi^{(k)}(x), &\text{if }x<b_k\end{array}\right.,
\end{align}
where we have used \eqref{eq:numeasure}, \eqref{function:u} and \eqref{first order condition} in the third line, and the fact that $\tilde{u}_0^{(k)}(x)>0$ if and only if $x<b_k$ in the last step.
\end{enumerate}

\subsection{Proof of Proposition \ref{prop:verify}}
For any fixed $x\in\R$, notice that $\{x+\xq\ge b_k^\star\}=\{\tau_{b_k^\star-x}^+\le \mathbf{e}_q\}$ except for a null set under measure $\pr^{\Phi(\alpha)}$ and that $X_{\tau_{b_k^\star-x}^+}=\overline{X}_{\tau_{b_k^\star-x}^+}$ on $\{\tau_{b_k^\star-x}^+<\infty\}$. Also, recall that $\tilde{u}_0^{(k)}(x)\ind_{\{x\ge b_k^\star\}}=\tilde{u}^{(k)}(x)$ for all $x\in\R$. We thus have
\begin{align}
&\lim_{q\downarrow0}\frac{\exph[\mte^{-\Phi(\alpha)\xq}(-\tilde{u}_0^{(k)}(x+\xq))\ind_{\{x+\xq\ge b_k^\star\}}]}{\exph[\mte^{-\Phi(\alpha)\xq}]}=\lim_{q\downarrow0}\frac{\exph[\mte^{-\Phi(\alpha)\xq}(-\tilde{u}_0^{(k)}(x+\xq))\ind_{\{\tau_{b_k^\star-x}^+\le \mathbf{e}_q\}}]}{\exph[\mte^{-\Phi(\alpha)\xq}]}\nn\\
=&\lim_{q\downarrow0}\frac{\exph[\exph[\mte^{-\Phi(\alpha)\xq}(-\tilde{u}_0^{(k)}(x+\xq))\ind_{\{\tau_{b_k^\star-x}^+\le \mathbf{e}_q\}}|\mathcal{F}_{\tau_{b_k^\star-x}^+}]]}{\exph[\mte^{-\Phi(\alpha)\xq}]}\nn\\
=&\lim_{q\downarrow0}\frac{\exph[\ind_{\{\tau_{b_k^\star-x}^+\le \mathbf{e}_q\}}\exph[\mte^{-\Phi(\alpha)(X_{\tau_{b_k^\star-x}^+}+\xq-X_{\tau_{b_k^\star-x}^+})}(-\tilde{u}_0^{(k)}(x+X_{\tau_{b_k^\star-x}^+}+\xq-X_{\tau_{b_k^\star-x}^+}))|\mathcal{F}_{\tau_{b_k^\star-x}^+}]]}{\exph[\mte^{-\Phi(\alpha)\xq}]}.\label{eqttt}
\end{align}

Now let us denote by $M_q:=\xq-X_{\tau_{b_k^\star-x}^+}$. Notice that $M_q\stackrel{\text{law}}{=}\xq$ and $M_q$ is independent of $\mathcal{F}_{\tau_{b_k^\star-x}^+}$ on the event $\{\tau_{b_k^\star-x}^+\le \mathbf{e}_q\}$, so  expression \eqref{eqttt} above is further equal to
\begin{align}
&\lim_{q\downarrow0}\frac{\exph[\ind_{\{\tau_{b_k^\star-x}^+\le \mathbf{e}_q\}}\mte^{-\Phi(\alpha)X_{\tau_{b_{k}^\star-x}^+}}\exph[\mte^{-\Phi(\alpha)M_q}(-\tilde{u}_0^{(k)}(x+X_{\tau_{b_k^\star-x}^+}+M_q))|X_{\tau_{b_k^\star-x}^+}]]}{\exph[\mte^{-\Phi(\alpha)\xq}]}\nn\\
=&\lim_{q\downarrow0}\frac{\exph[\mte^{-q\tau_{b_k^\star-x}^+-\Phi(\alpha)X_{\tau_{b_{k}^\star-x}^+}}\exph[\mte^{-\Phi(\alpha)M_q}(-\tilde{u}_0^{(k)}(x+X_{\tau_{b_k^\star-x}^+}+M_q))|X_{\tau_{b_k^\star-x}^+}]\ind_{\{\tau_{b_k^\star-x}^+<\infty\}}]}{\exph[\mte^{-\Phi(\alpha)\xq}]}\nn\\
=&\lim_{q\downarrow0}\ex\bigg[\mte^{-(q+\alpha)\tau_{b_k^\star-x}^{+}}\frac{\exph[\mte^{-\Phi(\alpha)M_q}(-\tilde{u}_0^{(k)}(x+X_{\tau_{b_k^\star-x}^+}+M_q))|X_{\tau_{b_k^\star-x}^+}]}{\exph[\mte^{-\Phi(\alpha)M_q}]}\ind_{\{\tau_{b_k^\star-x}^+<\infty\}}\bigg].\label{eq:do}
\end{align}
In the last equality, we have applied a change of measure, along with the dominated convergence theorem (see \eqref{value_at_a1}).

On the other hand, using the recursion \eqref{tildeu2} and mathematical induction we can show that there exist positive constants $C_1,C_2>0$ such that
\[|\tilde{u}_0^{(k)}(x)|\le C_1\mte^x+C_2,\,\,\,\forall x\in\R.\]
As a result,  the random variable 
$$\left|\mte^{-q\tau_{b_{k}^\star-x}^+}\frac{\exph[\mte^{-\Phi(\alpha)M_q}(-\tilde{u}_0^{(k)}(x+X_{\tau_{b_k^\star-x}^+}+M_q))|X_{\tau_{b_k^\star-x}^+}]}{\exph[\mte^{-\Phi(\alpha)M_q}]}\ind_{\{\tau_{b_k^\star-x}^+<\infty\}}\right|$$ 
is dominated by the  non-negative random variable
\begin{multline*}\mte^{-q\tau_{b_{k}^\star-x}^+}\frac{\exph[\mte^{-\Phi(\alpha)M_q}(C_1\mte^{x+X_{\tau_{b_k^\star-x}^+}+M_q}+C_2)|X_{\tau_{b_k^\star-x}^+}]}{\exph[\mte^{-\Phi(\alpha)M_q}]}\ind_{\{\tau_{b_k^\star-x}^+<\infty\}}\\
\le \bigg(C_1\mte^{x+X_{\tau_{b_{k}^\star-x}^+}}\frac{\exph[\mte^{-(\Phi(\alpha)-1)M_q}]}{\exph[\mte^{-\Phi(\alpha)M_q}]}+C_2\bigg)\ind_{\{\tau_{b_k^\star-x}^+<\infty\}}
\xrightarrow{\text{as }q\downarrow0}\bigg(C_1\mte^{x+X_{\tau_{b_{k}^\star-x}^+}}\psival+C_2\bigg)\ind_{\{\tau_{b_k^\star-x}^+<\infty\}},
\end{multline*}
where we used the fact that $\lim_{q\downarrow0}\frac{\exph[\mte^{-(\Phi(\alpha)-1)\xq}]}{\exph[\mte^{-\Phi(\alpha)\xq}]}=\psival$ in the last step.

Similarly, using \eqref{eq:vgE} we have  
\begin{align}
\phi^{(k)}(x)=&\mte^x-K+\ex[\mte^{-\alpha \delta}v^{(k-1)}(x+X_\delta)]\nn\\
=&\lim_{q\downarrow0}\exph\bigg[\frac{\mte^{-\Phi(\alpha)\xq}}{\exph[\mte^{-\Phi(\alpha)\xq}]}\bigg(\frac{\mte^{x+\xq}-\mte^{x_1^\star}}{\psival}-\ex[\mte^{-\alpha\delta}u^{(k-1)}(x+X_\delta+\xq)|\xq]\bigg)\bigg]\nn\\
=&\lim_{q\downarrow0}\frac{\exph[\mte^{-\Phi(\alpha)M_q}(-\tilde{u}_0^{(k)}(x+M_q))]}{\exph[\mte^{-\Phi(\alpha)\xq}]}.\label{phiEXP}
\end{align}
By the dominated convergence theorem, we obtain from \eqref{eq:do} and \eqref{phiEXP} that
\[\lim_{q\downarrow0}\frac{\exph[\mte^{-\Phi(\alpha)\xq}(-\tilde{u}_0^{(k)}(x+\xq))\ind_{\{x+\xq\ge b_k^\star\}}]}{\exph[\mte^{-\Phi(\alpha)\xq}]}=\ex_x\left[\mte^{-\alpha\tau_{b_k^\star}^+}\phi^{(k)}(X_{\tau_{b_k^\star}^+})\ind_{\{\tau_{b_k^\star}^+<\infty\}}\right]=g^{(k)}(x,b_k^\star).\]

To prove the supermartingale property, we use the fact that, on the event $\{t<\mathbf{e}_q\}$, we have $X_t+\sup_{s\in[t,\mathbf{e}_q]}(X_s-X_t)\le \xq$, $\pr^{\Phi(\alpha)}$-a.s. and $M_q:=\sup_{s\in[t,\mathbf{e}_q]}(X_s-X_t)$  has the same law as $\xq$, but is independent of $\mathcal{F}_t$.
 It follows from the non-negativity and the non-decreasing property of $-\tilde{u}^{(k)}$ that, for any $t>0$,
\begin{align*}
g^{(k)}(x,b_k^\star)=&\lim_{q\downarrow0}\frac{\exph[\mte^{-\Phi(\alpha)\xq}(-\tilde{u}^{(k)}(x+\xq))]}{\exph[\mte^{-\Phi(\alpha)\xq}]}\\
\ge&\lim_{q\downarrow0}\frac{\exph[\exph[\mte^{-\Phi(\alpha)\xq}(-\tilde{u}^{(k)}(x+\xq))\ind_{\{t<\mathbf{e}_q\}}|\mathcal{F}_t]]}{\exph[\mte^{-\Phi(\alpha)\xq}]}\\
\ge&\lim_{q\downarrow0}\frac{\exph[\exph[\mte^{-\Phi(\alpha)(X_t+M_q)}(-\tilde{u}^{(k)}(x+X_t+M_q))\ind_{\{t<\mathbf{e}_q\}}|\mathcal{F}_t]]}{\exph[\mte^{-\Phi(\alpha)\xq}]}\\
=&\lim_{q\downarrow0}\mte^{-qt}\frac{\exph[\exph[\mte^{-\Phi(\alpha)(X_t+M_q)}(-\tilde{u}^{(k)}(x+X_t+M_q))|\mathcal{F}_t]]}{\exph[\mte^{-\Phi(\alpha)\xq}]}\\
\ge&\exph\bigg[\mte^{-\Phi(\alpha)X_t}\cdot\lim_{q\downarrow0}\bigg(\frac{\exph[\mte^{-\Phi(\alpha)M_q}(-\tilde{u}^{(k)}(x+X_t+M_q))|X_t]}{\exph[\mte^{-\Phi(\alpha)M_q}]}\bigg)\bigg]\\
=&\ex[\mte^{-\alpha t}g^{(k)}(x+X_t,b_k^\star)]=\ex_x[\mte^{-\alpha t}g^{(k)}(X_t,b_k^\star)].
\end{align*}
Here the third inequality follows from Fatou's lemma and the Strong Markov property of $X$. This completes the proof.

\bibliographystyle{plain}
\singlespacing

\bibliography{ospbib}

\end{document}